\documentclass[11pt]{article}

\usepackage[margin=1in]{geometry}
\usepackage{amsmath,amssymb,amsthm,mathtools}
\usepackage{booktabs}
\usepackage{enumitem}
\usepackage{microtype}
\usepackage[colorlinks=true,linkcolor=blue,citecolor=blue,urlcolor=blue]{hyperref}
\usepackage{xcolor}
\usepackage{tikz,graphicx}
\newtheorem{theorem}{Theorem}
\newtheorem{lemma}{Lemma}

\theoremstyle{definition}
\newtheorem{definition}{Definition}
\newtheorem{remark}{Remark}

\newcommand{\E}{\mathbb E}

\newcommand{\Unif}{\operatorname{Unif}}

\providecommand{\OPT}{\operatorname{OPT}}
\providecommand{\Snap}{\operatorname{WSnap}}
\providecommand{\PsSnap}{\operatorname{WPsSnap}}

\providecommand{\norm}[1]{\left\lVert #1 \right\rVert}
\providecommand{\eps}{\varepsilon}
\providecommand{\sgnset}{\{+,-\}}
\providecommand{\clip}{\operatorname{clip}}
\providecommand{\Val}{\operatorname{Val}}
\providecommand{\bits}{\{0,1\}}
\newcommand{\1}{\mathbf 1}
\newcommand{\Qstate}[1]{\left|Q(#1)\right\rangle}
\newcommand{\inc}[1]{\mathrm{inc}(#1)}

\newcommand{\bx}{\mathbf{x}}
\newcommand{\br}{\mathbf{r}}

\usepackage{braket}
\providecommand{\Cut}{\operatorname{Cut}}
\providecommand{\MaxCut}{\operatorname{MaxCut}}

\newtheorem{claim}{Claim}

\title{Exponential Quantum Space Advantage for Approximating Max-$k$SAT in the Streaming Setting}
\author{
Haoyu Wang\thanks{Penn State University. Email: \texttt{hjw5492@psu.edu}.}\\
\and
Guangxu Yang\thanks{University of Southern California. Email: \texttt{guangxuy@usc.edu}. Research supported by NSF CAREER award 2141536.}\\
}
\begin{document}
\date{}
\maketitle

\begin{abstract}
In this paper, we give a one-pass quantum streaming algorithm for Max-$k$SAT that uses $\operatorname{polylog}(n)$ space and achieves a $0.7426$-approximation on instances with $n$ variables and $\operatorname{poly}(n)$ clauses. In contrast, prior work by Chou, Golovnev, and Velusamy (FOCS 2020) implies that achieving an approximation ratio better than $\sqrt{2}/2 \approx 0.7071$ for Max-$k$SAT requires $\Omega(\sqrt{n})$ space for any classical streaming algorithm. Therefore, it yields an exponential quantum space advantage for Max-$k$SAT in the streaming setting. Combining with the known results, it gives a complete classification of quantum space advantages for all Boolean Max-2CSPs. 

\end{abstract}
\newpage
\tableofcontents

\newpage
\section{Introduction}\label{sec:introduction}
Quantum computing offers the possibility of substantial resource advantages over classical computation. Quantum algorithms for combinatorial optimization have been studied intensively over the past three decades
\cite{FarhiEtAl2001Adiabatic,
FarhiGoldstoneGutmann2014BoundedOccurrence,
Hastings2018ShortPath,
BassoEtAl2022QAOAHighDepth,
DalzellEtAl2023MindTheGap,
KapitEtAl2023RandomHypergraphXORSAT,
ShaydulinEtAl2024ScalingAdvantage}.
This body of work has provided some evidence of possible superpolynomial quantum speedups for certain optimization problems
\cite{FarhiEtAl2025HighGirthMaxCut,
LengZhengWu2023NonconvexSeparation,
ChenLiuZhandry2022Filtering,
PirnayEtAl2024InPrincipleAdvantage,
Szegedy2022QuantumAdvantageSimplified,
GilyenHastingsVazirani2021AdiabaticAdvantage,
EldarHallgren2022LatticeApproximation,jordan2025optimization}.The main focus has traditionally been time complexity; for example, Shor's algorithm for integer factorization \cite{Shor1999} gives an exponential speedup over the best-known classical algorithms.   Nevertheless, establishing an unconditional superpolynomial quantum advantage in time complexity for optimization remains extremely challenging and is still largely open \cite{BernsteinVazirani1997, Watrous2018}. 

Space is another crucial resource, especially in the
quantum setting, where fault-tolerant qubits are scarce and costly \cite{FowlerMariantoniMartinisCleland2012, BabbushMcCleanGidneyBoixoNeven2021}. Streaming algorithms provide a natural framework for studying space-efficient quantum computation and for seeking quantum advantages in space usage \cite{LeGall2006, NayakTouchette2017, Kallaugher2022}. 
Motivated by processing massive datasets\cite{BabcockBabuDatarMotwaniWidom2002, Muthukrishnan2005}, in the streaming model, data elements arrive sequentially, and the algorithm must process each element as it appears while using only limited working space. Because of their space efficiency, streaming algorithms have found many applications in the computation of statistics over data streams \cite{FlajoletMartin1985, AlonMatiasSzegedy1999} and the estimation of parameters of massive graphs \cite{BarYossefKumarSivakumar2002, McGregor2014}.
Moreover, Niroula et al.~\cite{niroula2025realization} recently realized a quantum streaming algorithm on Quantinuum’s Helios trapped-ion quantum computer, using long-lived qubits that communicate with an external classical server. 

These theoretical and experimental developments make it particularly important to understand how broadly quantum space advantages extend across natural and practical streaming optimization problems. Recently, Kallaugher, Parekh, and Voronova achieved the first exponential quantum space advantage for a natural streaming problem \cite{KPV23}. They gave a polylogarithmic-space quantum streaming algorithm that achieves a $0.4844$-approximation for Max-DiCut. In contrast, previous work of Chou, Golovnev, and Velusamy~\cite{CGV20} implies that any classical streaming algorithm achieving an approximation ratio better than $4/9 \approx 0.444$ requires $\Omega(\sqrt n)$ space. 
Since Max-DiCut is only one particular Max-CSP, it naturally raises a broader question about the scope of quantum space advantages for streaming Max-CSPs. 
\begin{quote}
\textit{Which maximum constraint satisfaction problems (CSPs) admit an exponential quantum space advantage in the streaming model?}
\end{quote}
In this paper, we prove an exponential quantum space advantage for Max-$k$SAT in the streaming setting.  Max-\(k\)SAT is one of the most fundamental and widely studied problems in combinatorial optimization \cite{PapadimitriouYannakakis1991,Hastad2001OptimalInapproximability} and  it has played a central role in the development of approximation algorithms and hardness of approximation.  Our result also completes the classification of quantum space advantages for Boolean Max-2CSPs when combined with prior work.

\subsection{Our contributions}
Our main contribution is a polylogarithmic-space one-pass quantum streaming algorithm for Max-$k$SAT with approximation ratio $0.7426$.
We assume throughout that the number $m$ of clauses is polynomial in $n$.
\begin{theorem} \label{thm:main}
For every fixed $k\ge 2$ and every $\delta\in(0,1)$, there is a one-pass quantum streaming algorithm that, on every Max-$k$SAT instance with $n$ variables, outputs a $0.7426$-approximation with probability at least $1-\delta$, using $O\!\left(\log^5 n \log\frac{2}{\delta}\right)$ qubits of space.
\end{theorem}
In contrast, prior work by \cite{CGV20} implies that achieving an approximation ratio better than $\sqrt{2}/2 \approx 0.7071$ for Max-$k$SAT requires $\Omega(\sqrt{n})$ space for any classical streaming algorithm.  

\begin{figure}[t]
\label{figure_1}
  \centering
  \makebox[\linewidth][l]{%
  \hspace*{-5.69mm}%
  \resizebox{\linewidth}{!}{
\begingroup
\definecolor{qsBlue}{HTML}{1769AA}
\definecolor{qsOrange}{HTML}{BE661D}
\definecolor{qsInk}{HTML}{263340}
\definecolor{qsGray}{HTML}{78828C}
\definecolor{qsGuide}{HTML}{B7C1CA}
\definecolor{qsGrid}{HTML}{E9EDF0}
\begin{tikzpicture}[
  x=1cm, y=1cm, font=\small, text=qsInk,
  qs axis/.style={draw=qsGray,line width=.55pt},
  qs grid/.style={draw=qsGrid,line width=.4pt},
  qs guide/.style={draw=qsGuide,line width=.4pt},
  qs classical/.style={draw=qsOrange,line width=2.8pt,line cap=butt,line join=miter},
  qs quantum/.style={draw=qsBlue,line width=1.6pt,line cap=butt,line join=miter}
]
\def\xC{3.60}  
\def\xQ{8.28}  
\def\xT{9.84}  
\def\xR{12.00} 
\def\yLog{0.45}
\def\yPolylog{1.05}
\def\yRoot{4.10}
\def\yLinear{5.08}

\fill[qsBlue!12] (\xC,\yPolylog) rectangle (\xQ,\yRoot);
\foreach \yy in {\yLog,\yPolylog,\yRoot,\yLinear}
  \draw[qs grid] (0,\yy) -- (\xR,\yy);
\draw[qs guide] (\xC,0) -- (\xC,\yRoot);
\draw[qs guide] (\xQ,0) -- (\xQ,\yRoot);
\draw[qs guide] (\xT,0) -- (\xT,\yLinear);

\draw[qs axis] (0,5.75) -- (0,0) -- (\xR,0);
\foreach \yy in {2.35,2.49}
  \draw[white,line width=3pt] (-.07,\yy-.055) -- (.07,\yy+.055);
\foreach \yy in {2.35,2.49}
  \draw[qs axis,line width=.8pt] (-.07,\yy-.055) -- (.07,\yy+.055);

\draw[qs classical] (0,\yLog) -- (\xC,\yLog)
  -- (\xC,\yRoot) -- (\xT,\yRoot);
\draw[qs quantum] (\xC,\yPolylog) -- (\xQ,\yPolylog)
  -- (\xQ,\yRoot) -- (\xT,\yRoot)
  -- (\xT,\yLinear) -- (\xR,\yLinear);

\filldraw[fill=white,draw=qsOrange,line width=1.1pt]
  (\xC,\yLog) circle[radius=2.4pt];
\filldraw[fill=white,draw=qsOrange,line width=1.1pt]
  (\xC,\yRoot) circle[radius=2.4pt];
\filldraw[fill=white,draw=qsBlue,line width=1.1pt]
  (\xT,\yLinear) circle[radius=2.4pt];
\filldraw[fill=qsBlue,draw=white,line width=.7pt]
  (\xQ,\yPolylog) circle[radius=2.7pt];

\node[anchor=south west,align=left,text=qsOrange,inner sep=0pt]
  at (.20,.68) {Classical space\\$\Theta(\log n)$ bits};
\node[anchor=south,align=center,text=qsBlue,inner sep=0pt]
  at (5.94,1.25) {Quantum space (this work)\\$O(\log^5 n)$ qubits};
\node[anchor=south,text=qsOrange,inner sep=0pt]
  at (6.72,4.27) {Classical space: $\Theta(\sqrt{n})$ bits};
\node[anchor=south east,text=qsBlue,inner sep=0pt]
  at (\xR,5.26) {Quantum space: $\Theta(n)$ qubits};

\foreach \xx in {0,\xC,\xQ,\xT,\xR}
  \draw[qs axis] (\xx,0) -- (\xx,-.065);
\node[anchor=north] at (0,-.13) {$0$};
\node[anchor=north,align=center] at (\xC,-.13)
  {$\sqrt{2}/2\approx 0.7071$};
\node[anchor=north] at (\xQ,-.13) {$0.7426$};
\node[anchor=north] at (\xT,-.13) {$3/4$};
\node[anchor=north] at (\xR,-.13) {$1$};
\foreach \yy/\lab in {\yLog/{\log n},\yPolylog/{\log^5 n},\yRoot/{\sqrt{n}},\yLinear/{n}} {
  \draw[qs axis] (0,\yy) -- (-.065,\yy);
  \node[anchor=east,inner sep=0pt] at (-.15,\yy) {$\lab$};
}
\node[rotate=90,anchor=south] at (-1.05,2.85) {Space };
\node[anchor=north] at (6,-1.10)
  {Approximation ratio $\alpha$};
\end{tikzpicture}
\endgroup}%
}
  \caption{One-pass streaming space for Max-$k$SAT, for fixed $k\geq 2$
  and constant failure probability. The shaded region highlights the quantum
  advantage for $\sqrt{2}/2<\alpha\leq 0.7426$.
  Axis spacing and vertical connectors are schematic.
  For $0.7426<\alpha<3/4$, the quantum curve uses the classical
  $O(\sqrt{n})$ upper bound; the jump at $0.7426$ does not assert
  a quantum lower bound.}
  \label{fig:quantum-space-advantage}
\end{figure}

To better understand our result, we classify $2$-variable Boolean predicates $f \colon \{0,1\}^2 \to \{0,1\}$ into four standard types based on their truth tables: TR (depends on $\le 1$ input), OR (exactly one zero), XOR (depends on both inputs; exactly two zeros), and AND (exactly three zeros). For any $\Lambda \in \{\mathrm{OR}, \mathrm{XOR}, \mathrm{AND}\}$, the Max-2CSP problem is denoted $\text{Max-2E}\Lambda$ if its allowed predicate set $\mathcal{F}$ contains only $\Lambda$-type predicates (exact arity two), and $\text{Max-2}\Lambda$ if $\mathcal{F}$ also includes TR-type predicates. Identifying a predicate family with its contained types, every Boolean Max-2CSP is thus uniquely specified by a subset $\mathcal{F} \subseteq \{\mathrm{TR}, \mathrm{OR}, \mathrm{XOR}, \mathrm{AND}\}$.

Combining Theorem~\ref{thm:main} (take $k=2$)  with the known results \cite{CGV20, KP22} gives a complete classification of quantum space advantages for Boolean Max-2CSPs, summarized in Table~\ref{tab:max2csp-quantum-advantage}. Previous work had already settled several cases.  
\begin{itemize}
    \item For pure OR-type constraints (Max-2EOR) and XOR-type constraints (Max-Cut), quantum streaming algorithms satisfy the same space lower bounds as classical algorithms~\cite{KP22}, and hence give no quantum space advantage. 
    \item For AND-type constraints (Max-DiCut), Kallaugher, Parekh, and Voronova give a poly\-logarithmic-space quantum streaming algorithm achieving a $0.4844$-approximation~\cite{KPV23}, whereas the classical lower bound of Chou, Golovnev, and Velusamy rules out any ratio better than $4/9$ in $o(\sqrt n)$ space~\cite{CGV20}.
\end{itemize}  
Thus, only the TR+OR case, corresponding to Max-$2$OR / Max-$2$SAT, remained open. Our result closes this gap by showing that this problem also has a quantum space advantage in the streaming setting. 

\begin{table}[ht]
\centering
\small
\renewcommand{\arraystretch}{1.2}
\setlength{\tabcolsep}{5pt}
\begin{tabular}{p{0.18\textwidth}p{0.22\textwidth}p{0.16\textwidth}p{0.20\textwidth}p{0.12\textwidth}}
\toprule
Type $\mathcal{F}$
& Special case 
& Classical 
& Quantum 
& Advantage \\
\midrule
TR
& Trivial
& $1$ [Folklore]
& $1$ [Folklore]
& No \\

OR
& Max-2EOR
& $3/4$ \cite{CGV20}
& $3/4$ \cite{KP22}
& No \\

TR + OR
& Max-2SAT
& $\sqrt{2}/2$ \cite{CGV20}
& $\mathbf{[0.7426,3/4]}$ 
& \textbf{Yes}\\ 

XOR
& Max-Cut 
& $1/2$ \cite{CGV20}
& $1/2$ \cite{KP22}
& No \\ 

AND
& Max-DiCut 
& $4/9$ \cite{CGV20}
& $[0.4844,1/2]$ \cite{KPV23}
& Yes \\
\bottomrule
\end{tabular}
\caption{Classical and quantum one-pass streaming thresholds for Boolean Max-2CSPs with polylogarithmic space. The ``Classical'' and ``Quantum'' columns list the known threshold values from the cited works; an interval means that the quantum threshold is known to lie between the displayed lower and upper bounds. The last column indicates whether the quantum threshold is strictly better than the classical one. The bold entry is the new quantum advantage for Max-2SAT established in this work. Chou, Golovnev, and Velusamy~\cite{CGV20} show that the streaming approximability of every Boolean Max-2CSP is governed by the five canonical problems. More precisely, for any predicate family $\mathcal F$, $\operatorname{Max-CSP}(\mathcal F)$ is exactly as hard to approximate as the hardest problem in these predicates.}
\label{tab:max2csp-quantum-advantage}
\end{table}

To complement our results, we also prove a quantum streaming lower bound for Max-$k$SAT. This is achieved via a simple reduction from Max-Cut, leveraging its known quantum streaming space lower bound \cite{KP22}. The proof is deferred to Appendix~\ref{app:maxksat-lower-bound}.

\begin{theorem}[Quantum streaming lower bound for Max-$k$SAT] \label{thm:maxksat-quantum-lb}
Fix any $k\ge2$ and $\gamma>0$, any one-pass quantum streaming algorithm that achieves a $(0.75+\gamma)$-approximation for Max-$k$SAT requires $\Omega(n)$ qubits of space.
\end{theorem}

Intuitively, the quantum streaming advantage for Max-$k$SAT appears to rely on the non-uniformity of clause lengths. This is different from Max-DiCut, where all constraints are uniform arity-2 AND-type predicates. In contrast, for $k\ge 3$, Exact Max-$k$SAT admits a trivial $(1-2^{-k})$-approximation in logarithmic space via the random-assignment guarantee, so Theorem~\ref{thm:maxksat-quantum-lb} does not extend to the exact-arity setting. This separates Max-$k$SAT from Exact Max-$k$SAT in terms of quantum streaming space advantage.

\subsection{Related work}
\label{sec:prior-work}

\paragraph{Classical streaming bounds for Max-$k$SAT} The streaming approximability of CSPs has been studied for more than a decade; see the surveys and expository articles~\cite{Sud22, Ass23, Sin25}. Chou, Golovnev, and Velusamy \cite{CGV20} established a tight space complexity threshold for Max-$k$SAT around the approximation ratio of $\sqrt{2}/2 \approx 0.707$. Specifically, they demonstrated that a $(\sqrt{2}/2 - \varepsilon)$-approximation  streaming algorithm for Max-$k$SAT can be achieved using only $O(\varepsilon^{-2} \log n)$ space. On the negative side, they proved a matching lower bound: for any $k \ge 2$ and $\varepsilon > 0$, any streaming algorithm achieving a strictly better approximation ratio of $(\sqrt{2}/2 + \varepsilon)$ for Max-$k$-SAT requires $\Omega(\sqrt{n})$ space.

\paragraph{Quantum space advantages in streaming setting} 
The study of quantum space advantages in streaming began with Le Gall~\cite{LeGall2006}, who obtained an exponential separation in an online model with a prescribed arrival order.  For the standard streaming model, Gavinsky, Kempe, Kerenidis,
Raz, and deWolf~\cite{GavinskyKempeKerenidisRazDeWolf2008} established an $O(\log n)$-qubit versus $\Omega(\sqrt n)$-bit separation for a promise problem based on Boolean Hidden Matching.  Since these early problems were tailored to exhibit a separation, subsequent work sought quantum advantages for natural problems of independent classical
interest.  Several natural candidates such as $\mathrm{Dyck}(2)$ language \cite{JainNayak2014,NayakTouchette2017} and triangle counting \cite{Kallaugher2022} were subsequently studied. Kallaugher, Parekh, and Voronova achieved the first exponential quantum space advantage for a natural streaming problem \cite{KPV23}. They gave a polylogarithmic-space quantum streaming algorithm that achieves a $0.4844$-approximation for Max-DiCut. In contrast, previous work of Chou, Golovnev, and Velusamy~\cite{CGV20} implies that any classical streaming algorithm achieving an approximation ratio better than $4/9 \approx 0.444$ requires $\Omega(\sqrt n)$ space. 
In the multi-pass setting, Montanaro, Hamoudi and Magniez~\cite{Montanaro2016,HamoudiMagniez2019}
obtained quantum improvements for frequency-moment estimation. Feng, Xu, Li, Guo, and Lin~\cite{FengEtAl2026} recently established an exponential quantum space advantage for Shannon entropy estimation. 
Zhao et al.~\cite{zhao2026exponential} proved exponential quantum space advantages for classification and dimension reduction by processing random classical data samples sequentially.

\paragraph{Quantum advantages in communication complexity.}
Buhrman, Cleve, and Wigderson~\cite{BuhrmanCleveWigderson1998} established the first exponential quantum advantage in zero-error two-party communication.
Subsequent work established exponential separations in several other two-party models~\cite{Raz1999,BarYossefJayramKerenidis2008,KlartagRegev2011,commGavinsky16,commGavinsky19,commGavinsky20,commGRT22,commGGJL25}.
Gavinsky et al.~\cite{GavinskyKempeKerenidisRazDeWolf2008} established an exponential one-way separation based on Boolean Hidden Matching; our quantum primitive builds on their protocol.
Gavinsky~\cite{commGavinsky26} gave a total Boolean function with polylogarithmic quantum communication and a polynomial randomized lower bound.
Arunachalam, Girish, and Lifshitz~\cite{ArunachalamGirishLifshitz2023} demonstrated an exponential advantage using one clean qubit.
In number-on-forehead (NOF) communication, Gavinsky and Pudl\'ak~\cite{GavinskyPudlak2008} gave a separation for a relation under non-interactive communication.
Yang and Zhang obtained simultaneous~\cite{YangZhang2025Simultaneous} and one-way~\cite{YangZhang2026OneWay} separations for relations.
Wang and Wu~\cite{WangWu2026ConservativeLifting} showed a separation for a partial Boolean function in the  one-way conservative NOF model, where the last player only has a restricted view. 
Wang, Wu, and Yang~\cite{WangWuYang2026} obtained an exponential separation for an explicit partial Boolean function against unrestricted three-party NOF interaction, using one clean qubit on the quantum side.

\subsection{Technical overview}
Our proof extends the snapshot-to-pseudosnapshot framework of Kallaugher, Parekh, and Voronova \cite{KPV23}. We first construct a statistic by grouping and counting the  clauses, called a weighted snapshot; a linear combination of its coordinates gives a good approximation to the optimal value of Max-$k$SAT. 
Next,  we replace the snapshot by a pseudo-snapshot that can be estimated in the one pass streaming model.  Finally, we use quantum sketches to estimate the pseudo-snapshot and get a good approximation of  the optimal value. The following three steps give a high-level overview of our approach.

\paragraph{Step 1. From Max-$k$SAT to a weighted snapshot.}
We extend the snapshot approach of Saxena, Singer, Sudan, and Velusamy~\cite{saxena2023improved} from Max-DiCut to Max-$k$SAT. For an instance $\Phi$ with $m$ clauses, fix a length cutoff $\ell_0$ and write $\Phi_{\le\ell_0}$ for its short clauses. Weight their literal occurrences according to clause length. A variable's weighted bias is the difference between its positive and negative weighted degrees divided by its total weighted degree. We partition these biases into intervals, called buckets, and round each variable independently with a probability determined by its bucket. The cutoff, weights, buckets, and rounding probabilities form the rounding profile $\Theta$.

Label each literal by its sign and its variable's bucket. The weighted snapshot $\Snap(\Phi_{\le\ell_0})$ counts clauses of each length with a specified label at one position or specified labels at two positions; these are the  one-endpoint and two-endpoint coordinates of $\Snap(\Phi_{\le\ell_0})$. For clauses of length one or two, the expected number satisfied by rounding is a linear combination of these counts. Longer clauses depend on more than two labels, so computing their exact satisfaction probabilities would require querying more endpoints. Instead, for lengths $3\le\ell\le\ell_0$, we choose a linear combination of the same counts. This gives the short-clause score $L_{\le\ell_0}^\Theta(\Snap(\Phi_{\le\ell_0}))$.

Let $\rho$ be the target approximation ratio and $\OPT(\Phi)$ the maximum number of simultaneously satisfiable clauses. We choose the rounding probabilities so that every clause longer than $\ell_0$ is satisfied with probability at least $\rho$. If $M_{>\ell_0}$ is the number of these long clauses, our goal is to obtain the full score
\[
L^\Theta(\Phi)=L_{\le\ell_0}^\Theta(\Snap(\Phi_{\le\ell_0}))+\rho M_{>\ell_0},
\qquad
\rho\,\OPT(\Phi)\le L^\Theta(\Phi)\le\OPT(\Phi).
\]
For a fixed profile $\Theta$, a linear program $\mathsf{LP}(\Theta)$  chooses the score coefficients and maximizes $\rho$ subject to linear constraints ensuring these bounds.

We use 
an AI agent to  repeatedly adjust $\Theta$, solve the linear program $\mathsf{LP}(\Theta)$, and use the resulting ratio to guide the next adjustment. We verify the final profile $\Theta$ and  rational coefficients exactly, obtaining $\rho=0.742694813$. Figure~\ref{fig:proof-overview} illustrates this search and verification process; and we give more details about this AI agent numerical search at the end of the section.

\begin{figure}[htbp]
    \centering
    \makebox[\linewidth][l]{%
        \hspace*{-6.7mm}
        \resizebox{\linewidth}{!}{%
            \input{AI_agent.tikz}%
        }%
    }
    \caption{AI agent assisted parameter search and exact verification for the weighted snapshot.}
    \label{fig:proof-overview}
\end{figure}

\begin{samepage}
\paragraph{Step 2. The weighted snapshot and pseudo-snapshot are close.}
The labels above depend on final biases, which are unavailable when a clause arrives. Following Kallaugher, Parekh, and Voronova~\cite{KPV23}, we replace them by labels computed from approximate prefix degrees, counting occurrences up to the current clause, and exact suffix degrees, counting later occurrences. We use weighted total and positive degrees to form the corresponding pseudo-biases. Counting clauses by the resulting labels defines the weighted pseudo-snapshot $\PsSnap(\Phi_{\le\ell_0})$.

For a small precision parameter $\eps>0$, we prove that this replacement preserves the score guarantee up to $O(\eps m)$ additive error. Thus it remains to estimate $L_{\le\ell_0}^\Theta(\PsSnap(\Phi_{\le\ell_0}))$ within $\eps m$ in one pass.

\end{samepage}

\paragraph{Step 3. Quantum estimation of the pseudo-snapshot score.}
Kallaugher, Parekh, and Voronova~\cite{KPV23} maintain prefix degree information in superposition for each degree range and jointly query the two endpoints of each arriving edge. We adapt this idea to the one-endpoint and two-endpoint counts required by our score. Each degree level $a$ corresponds to an interval $(D_a,D_{a+1}]$; we maintain sketches for each level and each pair of levels.

As a clause arrives, threshold queries estimate whether its endpoints' prefix counts exceed specified thresholds. For a total prefix degree $d$, subtracting the upper-threshold indicator from the lower-threshold indicator selects a degree level:
\[
\1[D_a<d\le D_{a+1}]=\1[d>D_a]-\1[d>D_{a+1}].
\]
Pair-query outputs estimate products of two threshold indicators. Combining their outputs for the lower and upper thresholds therefore estimates the contribution from endpoints simultaneously satisfying the required conditions.

Once a query selects one or two variables from the current clause, the sketch counts their suffix degrees exactly during the rest of the stream. Combining these with the tested prefix values gives the corresponding pseudo-labels and hence the coordinate estimates. We sum over the degree levels, combine the estimates with the score coefficients, and average repeated sketches to obtain an estimate $\widehat L_{\le\ell_0}$ that, with constant probability, satisfies
\[
\left|\widehat L_{\le\ell_0}-L_{\le\ell_0}^\Theta(\PsSnap(\Phi_{\le\ell_0}))\right|\le\eps m.
\]
Adding $\rho M_{>\ell_0}$ and subtracting $O(\eps m)$ to account for both approximation and estimation errors gives a $(\rho-O(\eps))$-approximation, since $\OPT(\Phi)\ge m/2$. Choosing a sufficiently small constant $\eps$ gives the claimed $0.7426$-approximation. Independent repetitions and a median increase the success probability to $1-\delta$, using $O(\log^5 n\log(2/\delta))$ space.

\paragraph{AI agent  numerical search.}
An AI agent automated the search for the profile and coefficients  by repeatedly updating $\Theta$ and solving  $\mathsf{LP}(\Theta)$, reducing the need for manual parameter tuning.
As the number of buckets grew, constructing the full LP became impractical.
The agent instead solved an LP with selected constraints, found constraints violated by its solution, and added them before solving again.
For longer clauses, it grouped labels and shared two-endpoint coefficients within each pair of groups, while keeping separate one-endpoint coefficients for each label.
For these clauses, it checked all LP constraints without listing every label combination.
These changes supported a search with 1000 buckets.
An independent program verified the final rational coefficients exactly, yielding $\rho=0.742694813$.
\newcommand{\bw}{\mathbf{w}}

\section{Preliminaries}\label{sec:preliminaries}

 \paragraph{Max-CSPs.} Let $\mathcal F$ be a family of Boolean predicates $f:\bits^\ell\to\bits$, with $1\le\ell\le k$. An instance $\Phi=(C_1,\ldots,C_m)$ of Max-CSP$(\mathcal F)$ consists of $m$ constraints on Boolean variables, denoted by  $\bx:= (x_1,\ldots,x_n)$.

Each constraint applies a predicate in $\mathcal F$ to some literals, where a literal is a variable or its negation; we say a constraint is satisfied when the predicate evaluates to $1$. When a clause $C_i$, $i\in [m]$, is satisfied,  let $C_{i}(\bx)=1$; otherwise let $C_{i}(\bx)=0$.

An \textit{assignment} is a choice of a value in $\{0,1\}$ for each variable, denoted by $\bx \in\bits^n$. Given an assignment, define the number of satisfied constraints and the optimal one by 
\[
\Val_\Phi(\bx):=\sum_{t=1}^m C_t(\bx),\qquad
\OPT(\Phi):=\max_{\bx\in\bits^n}\Val_\Phi(\bx).
\]

Fix $0\le\rho\le1$ and $\delta\in(0,1)$.  $\mathcal{A}$ is a $\rho$-approximation algorithm for the Max-$k$SAT if its output $\mathcal{A}(\Phi)$ satisfies 
\[
\Pr\left[\rho\cdot \OPT(\Phi)\le \mathcal{A}(\Phi) \le\OPT(\Phi) \right]\ge1-\delta.
\] We call $\rho$ the approximation ratio.



\paragraph{Max-$k$SAT}. Max-$k$SAT is a specific case of Max-CSP$(\mathcal F)$ in which each constraint is a clause, i.e., the OR of at most $k$ literals. And the optimal value $\OPT(\Phi)$ is the maximum number of clauses satisfied over all assignments.

For ease of  description, we  \textit{simplify} the clauses:  remove repeated literals; count and discard clauses containing both a variable and its negation. Let $\Phi$ denote the remaining instance. Let $\Phi_{\le\ell_0}$ consist of the clauses of $\Phi$ of length at most $\ell_0$. Then we add the discarded tautological clauses back to the final count directly.

Unless otherwise stated, we assume throughout that all clauses are simplified.

\paragraph{Quantum streaming algorithms.}
In this paper, we design an approximation algorithm $\mathcal{A}$ of $\OPT(\Phi)$, more specifically, a one-pass quantum streaming algorithm. For the input instance $\Phi = (C_1,...,C_m)$, it processes the constraints left-to-right using a limited quantum workspace.

First, the algorithm receives the number of variables $n$ and the stream length $m$. Then it initializes a quantum work register $\mathcal W$  by an initial state $\varrho_0$.

The algorithm receives the constraints $C_1, \ldots, C_m$ one by one (in a possibly adversarial order). For each constraint $C_i$, $i\in [m]$, the algorithm applies a quantum channel $\mathcal H_{C_i}$ acting on $\mathcal W$. As the input stream is processed, the state of the work register evolves as:
\[
  \varrho_i =  \mathcal H_{C_i}(\varrho_{i-1}), \qquad i=1,2,\ldots,m.
\]
    
After the stream, the algorithm measures the register $\mathcal W$ and outputs a value $\mathcal{A}(\Phi)$.

In this quantum streaming framework, the space complexity of the algorithm is the number of qubits used in the register $\mathcal{W}$. We place no restriction on the running time or computational complexity of the quantum update channels. We also allow the algorithm read-only access to public random bits, and these random bits are not charged to the working space.

\paragraph{Random rounding.} The most natural way to find an assignment with average behavior is random rounding, i.e., for each variable $x_v$, $v\in [n]$, define a rounding probability $r_v\in [0,1]$ and choose its value independently so that 
\[
\Pr[x_v=1] = r_v,\quad \Pr[x_v=0] = 1-  r_v.
\]
Denote the sign of  literal by $s\in \{+,-\}$. Thus for each literal $(x_v,s)$, the probability that this literal is satisfied is
$r_v$ if $s= +$ while $1-r_v$ if $s= -$.

For a clause $C$ and its $j$-th literal, denote the variable index by $v_{j}(C)\in [n]$ and the sign by $s_{j}(C)\in \{+,-\}$.
Define the satisfaction probability of the $j$-th literal by \[
q_{j} := \begin{cases}
    r_{v_{j}(C)},& \text{ if }\;\; s_{j}(C)=+,\\
    1-r_{v_{j}(C)},& \text{ if }\;\; s_{j}(C)=-.
\end{cases}
\]
Thus the probability that a clause is satisfied is determined by the rounding probability and the signs in each position, i.e., 
\begin{equation}
\label{eq:satisfied_probability}
    1-\prod_{j=1}^{|C|}(1-q_{j}).
\end{equation} 

Under random rounding, $\E[\Val_\Phi(\bx)]\le\OPT(\Phi)$. To obtain $\E[\Val_\Phi(\bx)]\ge\rho\OPT(\Phi)$, we need to carefully design the rounding probabilities.

\paragraph{Bias, buckets, labels, and  rounding probability.}
For a variable $x_v$, let the \textit{positive degree}  $p_v$ be the total  number of the occurrences of $x_v$ and let \textit{negative degree}  $q_v$  be the total number of the occurrences of $\lnot x_v$,  among all clauses. And let the  \textit{total degree}  be $d_v:=p_v+q_v$. Therefore, we can define the \textit{bias} of $x_v$ by  \[
b_{v}:=\begin{cases}
    (p_v-q_v)/d_v, & \text{ if }\;\; d_v>0,\\
    0, & \text{ if }\;\; d_v=0.
\end{cases} 
\]

Bias is a natural guide for designing the rounding probability. Notice that the bias $b_v\in [-1,1]$. If $b_v\to 1$, it means most literals on $x_v$ have positive sign, so we should choose the rounding probability $r_v\to 1$; similarly, if $b_v\to -1$, we should let $r_v\to 0$. Moreover, storing an independent $r_v$ for each $x_v$, $v\in [n]$, will cost $\Omega(n)$ memory space; thus to achieve the logarithmic space, we group the variables by bias and let the variables in a  group share the same rounding probability.

More concretely, we partition  $[-1,1]$ into $L$ intervals, denoted by $\mathcal{I}:= \{I_1,I_2,...,I_L\}$ where  \[
  I_i=[a_i,b_i)\quad(1\le i\le L-1),\qquad 
  I_{L}=[a_{L},b_{L}],
  \qquad \text{ and }\;\; [-1,1]=\bigsqcup_{i=1}^{L} I_{i}.
\]
We call $\{I_1,I_2,...,I_L\}$ \textit{buckets}. For each bucket $I_{i}$, we define a rounding probability $r_i\in [0,1]$. Let the \textit{rounding vector} $\br= (r_1,\ldots,r_L)$.

Given a clause $C$ and the $j$-th literal with bias  in $I_i$ and sign $s\in \{+,-\}$, we define the \textit{label} \footnote{We use the basic version of bias here for ease of explanation; later we will replace the bias condition by a more well-designed version called  weighted bias.} $$\alpha_{j} (C) := (s,i)\in \{+,-\}\times [L].$$ Thus  the satisfaction probability is \[
q_{\alpha_{j}(C)} := \begin{cases}
    r_{i},& \text{ if }\;\; s_{j}(C)=+,\\
    1-r_{i},& \text{ if }\;\; s_{j}(C)=-.
\end{cases}
\]
Sometimes we write $q_{\alpha}$ and  $\alpha\in \{+,-\}\times [L]$ for short, if the position $j$ and clause $C$ are clear from context or no particular choice is intended.

\paragraph{Weighted bias.}
A key observation is that $\Pr[C(\bx)=1]\ge1-(\max_\alpha(1-q_\alpha))^{|C|}$. Fix a length threshold $\ell_0$ such that $1-(\max_\alpha(1-q_\alpha))^{\ell_0+1}\ge\rho$. Then every clause of length greater than $\ell_0$ is satisfied with probability at least $\rho$. 
Therefore, in the following argument, we focus on the  clauses with length at most $\ell_0$, denoted by $\Phi_{\le \ell_0}$.

To incorporate the effect of clause length into the bias, we define the \textit{weighted bias}. Throughout the rest of the paper, we use these weighted versions and may omit the qualifier ``weighted''.

\begin{definition}[Weighted bias]
\label{def:maxksat-bias}
Fix a variable $x_v$ and  length $\ell\in [\ell_0]$. The \textit{positive degree} $p_v^{(\ell)}$ counts clauses of length $\ell$ containing $x_v$. The \textit{negative degree} $q_v^{(\ell)}$ counts those containing $\neg x_v$. With weights $\bw := (w_1,w_2,...,w_{\ell_0})\in\mathbb Q_{>0}^{\ell_0}$, we define the weighted positive and negative  degree by 
\[
   \widetilde{p }_v:= \sum_{\ell=1}^{\ell_0} w_{\ell}\cdot p_{v}^{(\ell)}  ,\qquad
  \widetilde q_v:= \sum_{\ell=1}^{\ell_0} w_{\ell}\cdot q_{v}^{(\ell)}. 
\]
The \textit{weighted bias} is
\[
  \widetilde b_v:=\begin{cases}
    ( \widetilde p_v-  \widetilde q_v)/( \widetilde p_v+ \widetilde q_v),& \widetilde p_v+ \widetilde q_v>0,\\
    0,& \widetilde p_v+ \widetilde q_v=0.
  \end{cases}
\]
\end{definition}

Given the weighted bias, the satisfaction probability is  $q_{\alpha}$ on  label    $\alpha = (s,i)$ where  $i$ denotes the bucket index with $\widetilde b_v\in I_{i}$.

\section{Reducing \texorpdfstring{Max-$k$SAT}{Max-kSAT} to weighted snapshot estimation}
\label{sec:weighted-snapshot-reduction}

We follow the snapshot approach of Saxena, Singer, Sudan, and Velusamy~\cite{saxena2023improved} on Max-DiCut to reduce the Max-$k$SAT to a statistic  estimation task.

The discussion above reduces our task to choosing suitable length threshold $\ell_0$, weights $\bw$,  buckets $\mathcal{I}$, and the  rounding vector $\br$ so that the resulting randomized rounding achieves the desired approximation ratio. For brevity, we call them the \textit{rounding profile}  and denote it by $\Theta: =(\ell_0, \bw,\mathcal{I} ,\br)$; we defer the discussion of finding good rounding profile to Appendix~\ref{sec:Snapshot_score_approximation}.

Suppose $\Theta =(\ell_0, \bw,\mathcal{I} ,\br)$ is given, we can define the \textit{weighted snapshot}, which groups and counts the clauses by their \text{length} and the \text{label} of each literal; it  serves as the key quantity for our algorithm design.
\begin{definition}[Weighted snapshot]
\label{def:weighted-snapshot}
    The weighted snapshot, denoted by $\Snap(\Phi_{\le \ell_0})$, consists of two kinds of quantities; we call them one-endpoint and two-endpoint coordinates.   The
    \textit{one-endpoint coordinate} counts the clauses satisfying length $\ell$ and the $j$-th literal has label $\alpha\in \{+,-\}\times [L]$ as   \[
  U_{\ell,j,\alpha}(\Phi):=\#\{C\in\Phi:|C|=\ell,\ \alpha_j(C)=\alpha\}
  \qquad(1\le j\le \ell\le \ell_0),
\]
and the \textit{two-endpoint coordinate}  counts the clauses satisfying  the respective requirements on clause length and  labels at the two positions, i.e., for $\alpha,\beta \in  \{+,-\}\times [L]$ 
\[
  P_{\ell,j,t,\alpha,\beta}(\Phi):=
  \#\{C\in\Phi:|C|=\ell,\ \alpha_j(C)=\alpha,\ \alpha_t(C)=\beta\}
  \qquad(1\le j<t\le\ell\le\ell_0),
\]
\end{definition}

\medskip
Round each variable independently according to $\Theta$ and its weighted bias in $\Phi_{\le\ell_0}$, obtaining an assignment $\bx$. For every subinstance $\Psi\subseteq\Phi$, define $\mathcal E^\Theta(\Psi):=\E[\Val_\Psi(\bx)]$ using this same $\bx$.
Given the weighted snapshot, 
the expected number of satisfied clauses of  $\Phi_{\le 2}$ can be written as  \begin{equation}
  \label{eq:expect_2_term}
  \mathcal{E}^{\Theta}(\Phi_{\le 2}) = \sum_\alpha U_{1,1,\alpha}q_\alpha +\sum_{\alpha,\beta}P_{2,1,2,\alpha,\beta} \bigl(1-(1-q_\alpha)(1-q_\beta)\bigr).
\end{equation}

For the clauses with length  $3\le \ell\le \ell_0$, we use the linear combination of the one-endpoint and two-endpoint coordinates as approximate substitutes. Thus we can define the \textit{snapshot score}, which bridges the expected number of satisfied clauses and efficient algorithm design.

\begin{definition}[Snapshot score]
\label{def:snapshot-score}
Let $M_{\ell}$ be  the number of clauses with length $\ell$. 
Given the rounding profile  $\Theta:= (\ell_0,\bw,\mathcal{I},\br)$, define the snapshot score of  $\Phi_{\le \ell_0}$ by  
    \[
\begin{aligned}
  L_{\le\ell_0}^\Theta(\Snap(\Phi_{\le\ell_0}))
  := &\sum_\alpha U_{1,1,\alpha}q_\alpha
    +\sum_{\alpha,\beta}P_{2,1,2,\alpha,\beta}
      \bigl(1-(1-q_\alpha)(1-q_\beta)\bigr)\\
  &+\sum_{\ell=3}^{\ell_0} \bigg[   c_{\ell}\cdot M_{\ell} +\sum_{1\le j\le \ell} \sum_\alpha u_{\ell,j}(\alpha)U_{\ell,j,\alpha}
    +\sum_{1\le j<t\le\ell} \sum_{\alpha,\beta}
      p_{\ell,jt}(\alpha,\beta)P_{\ell,j,t,\alpha,\beta} \bigg],
\end{aligned}
\]
where $q_{\alpha}$,  $q_{\beta}$ are the satisfaction probability from $\Theta$ and  $c_{\ell},u_{\ell,j}(\alpha), p_{\ell,jt}(\alpha,\beta )\in \mathbb{R}$ are linear   coefficients to be determined. 
Moreover, given an approximation ratio $\rho\in [0,1]$, let the snapshot  score of $\Phi$ be 
\[
L^\Theta(\Phi )  = L_{\le\ell_0}^\Theta(\Snap(\Phi_{\le\ell_0})) + \rho\cdot M_{> \ell_0},
\]
where $M_{> \ell_0} := \sum_{\ell=\ell_0 +1}^{n} M_{\ell}$.
\end{definition}

\begin{lemma}
\label{lem:maxksat-weighted-snapshot-bound}
 There exist  a rounding profile $\Theta= (\ell_0,\bw,\mathcal{I},\br)$ and coefficients $ \{c_{\ell
 }, u_{\ell,j}(\alpha), p_{\ell,jt}(\alpha,\beta) \}$ where $3\le \ell\le \ell_0$, $j\in[\ell]$ for $u_{\ell,j}$, $1\le j<t\le\ell$ for $p_{\ell,jt}$, and $\alpha,\beta\in \{+,-\}\times [L]$ so that  for $\rho = 0.742694813$,  we have
  \begin{equation}
  1-\left(\max_{\alpha}(1-q_\alpha)\right)^{\ell_0+1}\ge\rho
  \label{eq:snapshot-long-clause-bound},
 \end{equation}
 and
 \[
 \rho\cdot \OPT(\Phi_{\le \ell_0})\le  L_{\le\ell_0}^\Theta(\Snap(\Phi_{\le\ell_0})) \le \mathcal{E}^{\Theta}(\Phi_{\le\ell_0}).
 \]
\end{lemma}

The proof of Lemma~\ref{lem:maxksat-weighted-snapshot-bound} is in Appendix~\ref{sec:Snapshot_score_approximation}; it gives the following approximation guarantee for the optimal value of Max-$k$SAT.

\begin{theorem}[Max-$k$SAT to snapshot score] 
\label{thm:maxksat_to_score}
There exist a  rounding profile $\Theta$ and coefficients $\{c_{\ell}, u_{\ell,j}(\alpha), p_{\ell,jt}(\alpha,\beta)\}$ where  $3\le \ell\le \ell_0$, $j\in[\ell]$ for $u_{\ell,j}$, $1\le j<t\le\ell$ for $p_{\ell,jt}$, and $\alpha,\beta\in \{+,-\}\times [L]$ so that for  $\rho = 0.742694813$, we have 
\[
  \rho\,\OPT(\Phi)\le L^{\Theta}(\Phi)\le\OPT(\Phi).
\]

\end{theorem}

\begin{proof}[Proof of Theorem~\ref{thm:maxksat_to_score}]
Take the rounding profile and coefficients selected in  Lemma~\ref{lem:maxksat-weighted-snapshot-bound}.
Therefore, 
\[
 L^\Theta(\Phi )  = L_{\le\ell_0}^\Theta(\Snap(\Phi_{\le\ell_0})) + \rho M_{> \ell_0} 
  \ge\rho\,\OPT(\Phi_{\le\ell_0})+\rho M_{>\ell_0}
  \ge\rho\,\OPT(\Phi).
\]
For the upper bound, apply the independent random rounding by $\Theta$; since the variables in the  simplified clause are distinct,  every clause $C\in \Phi$ with $|C|\ge \ell_0+1$  satisfies
\[
  \E[C(\bx)]
  =1-\prod_{j=1}^{|C|}(1-q_{\alpha_j(C)})
  \ge1-\left(\max_{\alpha}(1-q_\alpha)\right)^{\ell_0+1}
  \ge\rho.
\]
Hence
\[
\sum_{C:\,|C|\ge \ell_0+1}\E[C(\bx)] \ge \rho M_{>\ell_0}.
\]
Together with the  upper bound of the snapshot score in Lemma~\ref{lem:maxksat-weighted-snapshot-bound}, this gives
\[
  L^\Theta(\Phi)
  \le\mathcal E^\Theta(\Phi_{\le\ell_0})+\sum_{C:\,|C|\ge \ell_0+1}\E[C(\bx)]
  =\mathcal E^\Theta(\Phi)\le\OPT(\Phi). \qedhere
\] 
\end{proof}

\section{Closeness of the weighted snapshot and pseudo-snapshot}
\label{sec:weighted-pseudosnapshot-closeness}

The last  section reduces approximating $\OPT(\Phi)$ to estimating the weighted snapshot score $L^{\Theta}(\Phi)$.
With the rounding profile and coefficients fixed, the remaining task is to estimate the one-endpoint and two-endpoint coordinates of the weighted snapshot.
The weighted snapshot groups and counts clauses by length, signs, and bias; the first two can be verified immediately as each clause arrives, whereas  the bias cannot: it can only be computed after the stream ends, once all degree information has been collected. However, with only logarithmic space, we can store at most \(O(1)\) clauses, which is insufficient to maintain the statistic over \(\operatorname{poly}(n)\) clauses.

In this section, we extend the pseudo-bias and pseudo-snapshot construction of Kallaugher, Parekh, and Voronova~\cite{KPV23} to the weighted degrees of Max-$k$SAT.
When each clause \(C\) arrives, we split the degree information at \(C\) into prefix and suffix parts.
Our quantum algorithm updates prefix degrees as clauses arrive; it  samples clauses with the required length and literal signs with probability $\tau$, stores those selected, and counts their exact suffix degrees.
After the stream, each sampled clause whose pseudo-biases lie in the required buckets contributes $1/\tau$ to the estimate.

We first define the prefix and suffix degrees,  together with the corresponding approximations; then we  define the pseudo-bias and the  pseudo-snapshot. 
Finally we show that using pseudo-bias as the replacement of the bias only costs an additive $O(\eps m)$ loss.

\paragraph{Prefix and suffix degrees.}
Throughout this section  we focus on the short instance $\Phi_{\le \ell_0}$, since for the long clause we count the number $M_{>\ell_0}$ directly.
For time $t\ge 1$, denote the clause arriving at $t$ by $C=C_t$; and for any  $i\in [t]$, let $\ell_{i}:= |C_i|$. For any variable $x_v$,  $v\in [n]$, define the  prefix positive degree and prefix  total degree by
\[
  p_v^{\le C}
  :=\sum_{\substack{i\in[t],\ \ell_{i} \le\ell_0\\ x_v\in C_i}}w_{\ell_i },
  \qquad\qquad
  d_v^{\le C} :=\sum_{\substack{i\in[t],\ \ell_{i} \le\ell_0\\
  x_v\in C_i\ \text{or}\ \neg x_v\in C_i}}w_{\ell_i },
\]
where $w_{\ell_i}$ is the corresponding entry of $\bw$; see Definition~\ref{def:maxksat-bias}.
Recall that $\widetilde p_v$ and $\widetilde q_v$ are the positive weighted degree and the negative weighted degree.
Let $\widetilde d_v:=\widetilde p_v+\widetilde q_v$, and define the corresponding suffix degrees by
\[
  p_v^{>C}:=\widetilde p_v-p_v^{\le C},
  \qquad
  d_v^{>C}:=\widetilde d_v-d_v^{\le C}.
\]

\paragraph{Literal copies.}
Multiply the rational weights by a common denominator so that $w_\ell\in\mathbb Z_{>0}$; this leaves all biases and snapshot scores unchanged.
For ease of  algorithm design, we count degrees by making $w_{|C|}$ copies of each relevant literal in a clause $C$, rather than multiplying by $w_{|C|}$; and we apply the same tests and counting procedure to each copy.

For each clause $C\in\Phi_{\le\ell_0}$, denote the length  $\ell:= |C|$ and  make $w_{\ell}$ copies of each literal. For the $j$-th literal of $C$,   define 
the copy $(C,j,h)$ with variable $v_j(C)$ and sign $s_j(C)$, for each  $j\in[\ell]$ and $h\in[w_{\ell } ]$.
The prefix and suffix degrees above count these copies; the algorithm does not store the copies separately.
In particular, the total number of literal copies in $\Phi_{\le\ell_0}$ is
\begin{equation}
  \sum_{v=1}^n\widetilde d_v
  =\sum_{\ell=1}^{\ell_0}\ell w_\ell M_\ell
  \le m\cdot \max_{\ell\in[\ell_0]}\ell w_\ell = O(m),
  \label{eq:weighted-total-copies}
\end{equation}
where $\max_{\ell\in[\ell_0]}\ell w_\ell = O(1)$ by our parameter choices.

\paragraph{Degree intervals.} For our quantum algorithm in the next section,  
the quantum estimators test whether the total prefix degree lies in a specified degree interval; 
and  we approximate it by the interval  endpoint; besides, we  approximate the positive prefix degree by sampling positive copies and capping the sampled count.

More concretely,  fix $0<\eps<1/10$. Define the degree levels by a geometric sequence as
\[
  D_0:=0,\qquad D_a:=\lfloor(1+\eps^3)^a\rfloor\quad(1\le a\le A),
\]
where $A$ is the least index with $D_A\ge m\max_{\ell\in[\ell_0]}\ell w_\ell$; by Inequality~\eqref{eq:weighted-total-copies}, $A=O_\eps(\log(m+1))$.
The \textit{degree intervals} are $(D_a,D_{a+1}]$ for $0\le a<A$; empty intervals are ignored.

\paragraph{Degree  approximation by random sampling.} 
Since the prefix  degree may  range from $1$ to $O(m)$; to save the memory space, we use random sampling to sketch the high degree information so that the algorithm can only spend $O(1)$ space on the degree information of each variable.

Let $\kappa = \lceil K\eps^{-7} \rceil $ where $K>0$ is a sufficiently large constant.
For each degree interval $0\le a<A$ with $2D_a\ge\kappa$ and each positive copy $(C,j,h)$,
define a random bit
\begin{equation}
    \label{eq:random_sampling}
    f_a(C,j,h)\in\bits,\qquad \Pr[f_a(C,j,h)=1] =\frac{\kappa}{2D_a}.
\end{equation}

These sampling bits, denoted collectively by $\{f_a\}$, are mutually independent over all intervals and copies.

For a clause $C$ containing $x_v$, choose $a$ with $D_a<d_v^{\le C}\le D_{a+1}$.
If $2D_a<\kappa$, keep the prefix degrees exact:
\[
  \widetilde d_v^{\le C}:=d_v^{\le C},\qquad
  \widetilde p_v^{\le C}:=p_v^{\le C}.
\]
Otherwise, let $s_v^{\le C}$ count the positive copies $(C',j,h)$ of $x_v$ in the prefix through $C$ for which $f_a(C',j,h)=1$, and set
\[
  \widetilde d_v^{\le C}:=D_a,\qquad
  \widetilde p_v^{\le C}:=\frac{2D_a}{\kappa}\min\{s_v^{\le C},\kappa\}.
\]

\paragraph{Random noise.}     
The random sampling sketch will cost an  additive error of the degree and the bias; and 
an error in bias may change its bucket index. 
To control these changes, independently sample a noise \begin{equation}\label{eq:noise}
g(v)\sim\Unif[-\eps,\eps]   
\end{equation}
for each $v\in[n]$.
The same $g(v)$ is used in every clause containing $x_v$.
Write $$\clip(z):=\max\{-1,\min\{1,z\}\}.$$

\medskip
We now define the bias and snapshot actually used by our algorithm, termed pseudo-bias and pseudo-snapshot.
\begin{definition}[Weighted pseudo-bias]
\label{def:weighted-pseudobias}
For a clause $C\in\Phi_{\le\ell_0}$ and a variable $x_v$ occurring in $C$, define its weighted pseudo-bias by
\[
  \widetilde b_v^{\,C}:=
  \clip\left(
    \frac{2(\widetilde p_v^{\le C}+p_v^{>C})}
         {\widetilde d_v^{\le C}+d_v^{>C}}
    -1+g(v)
  \right).
\]
\end{definition}

We use the pseudo-bias in place of the bias when assigning labels; and define the weighted pseudo-snapshot and the   score as follows.
\begin{definition}[Weighted pseudo-snapshot]
\label{def:weighted-pseudosnapshot}
For any  clause $C\in\Phi_{\le\ell_0}$ with length $|C|=\ell$ and the position   $j\in[\ell ]$, let $\widetilde i_j(C)\in [L]$ be the bucket index with pseudo-bias $\widetilde b_{v_j(C)}^{\,C}\in I_{\widetilde i_j(C)}$, and define the label
\[
  \widetilde\alpha_j(C):=(s_j(C),\widetilde i_j(C)).
\]
The weighted pseudo-snapshot, denoted by $\PsSnap(\Phi_{\le\ell_0})$, consists of the one-endpoint and two-endpoint coordinates
\[
  \widetilde U_{\ell,j,\alpha}(\Phi)
  :=\#\{C\in\Phi:|C|=\ell,\ \widetilde\alpha_j(C)=\alpha\},
\]
\[
  \widetilde P_{\ell,j,t,\alpha,\beta}(\Phi)
  :=\#\{C\in\Phi:|C|=\ell,\ \widetilde\alpha_j(C)=\alpha,\ \widetilde\alpha_t(C)=\beta\},
\]
where $1\le j\le\ell\le\ell_0$ and $\alpha,\beta\in\sgnset\times[L]$, with $j<t\le\ell$ for the two-endpoint coordinates.
\end{definition}

\begin{definition}
[Weighted pseudo-snapshot score]
\label{def:weighted-pseudosnapshot-score} 
Use the rounding profile and coefficients giving the snapshot score approximation in Theorem~\ref{thm:maxksat_to_score}, keep $M_\ell$ unchanged, and define the weighted pseudo-snapshot score as
\[
\begin{aligned}
  L_{\le\ell_0}^\Theta(\PsSnap(\Phi_{\le\ell_0}))
  := &\sum_\alpha \widetilde U_{1,1,\alpha}q_\alpha
    +\sum_{\alpha,\beta}\widetilde  P_{2,1,2,\alpha,\beta}
      \bigl(1-(1-q_\alpha)(1-q_\beta)\bigr)\\
  &+\sum_{\ell=3}^{\ell_0} \bigg[   c_{\ell}\cdot M_{\ell} +\sum_{1\le j\le \ell} \sum_\alpha u_{\ell,j}(\alpha) \widetilde  U_{\ell,j,\alpha}
    +\sum_{1\le j<t\le\ell} \sum_{\alpha,\beta}
      p_{\ell,jt}(\alpha,\beta) \widetilde  P_{\ell,j,t,\alpha,\beta} \bigg],
\end{aligned}
\]
\end{definition}

The following theorem shows that the pseudo-snapshot score approximates $\OPT(\Phi)$ with a small extra loss. Recall the sampling bits $f_a$ for positive literal copies in~\eqref{eq:random_sampling} and the uniform noise $g(v)\sim\Unif[-\eps,\eps]$ in~\eqref{eq:noise}, shared by all clauses containing $x_v$.

\begin{theorem}[Weighted pseudo-snapshot to value]
\label{thm:weighted-pseudosnapshot-to-value}
There exist constants $\gamma_0,\gamma_1>0$, depending only on the fixed rounding profile and coefficients, so that with probability at least $1-O(\eps)$ over the sampling bits $\{f_a\}$ and noise $g$, we have
\[
  \rho\,\OPT(\Phi)-\gamma_1\eps m
  \le L_{\le\ell_0}^\Theta(\PsSnap(\Phi_{\le\ell_0}))
     +\rho M_{>\ell_0}-\gamma_0\eps m
  \le\OPT(\Phi).
\]
\end{theorem}
The proof of Theorem~\ref{thm:weighted-pseudosnapshot-to-value} is deferred to the end of this  section. Before proving Theorem~\ref{thm:weighted-pseudosnapshot-to-value}, we need the following two lemmas; the first lemma shows   that the weighted snapshot score with under random noise can also  approximate $\OPT(\Phi)$ with an additive error, and the second lemma shows that the weighted pseudo-snapshot is close to the weighted snapshot.

\medskip
Apply the same variable noise to the weighted biases and set
\[
  b'_v:=\clip(\widetilde b_v+g(v)),\qquad
  b'_v\in I_{i'(v)}\quad(v\in[n]).
\]
Let $\Snap_{\eps}(\Phi_{\le\ell_0})$ be the weighted snapshot obtained by replacing each label $\alpha_j(C)$ with $(s_j(C),i'(v_j(C)))$ in Definition~\ref{def:weighted-snapshot}.
Since $\widetilde b_v\in[-1,1]$ and $|g(v)|\le\eps$, we have
\[
  |b'_v-\widetilde b_v|\le\eps
  \quad\Longrightarrow\quad
  a_{i'(v)}-\eps\le\widetilde b_v\le b_{i'(v)}+\eps.
\]
The following lemma shows that these bucket indices retain the approximation guarantee with an additive error.
\begin{samepage}
\begin{lemma}[Snapshot score with an additive error]
\label{lem:weighted-snapshot-score-error}
Use the rounding profile $\Theta$ and coefficients giving the snapshot score approximation in Lemma~\ref{lem:maxksat-weighted-snapshot-bound}.
There is a constant $\gamma_2>0$, depending only on the fixed rounding profile and coefficients, such that for every choice of $g(v)\in[-\eps,\eps]$, $v\in[n]$, we have
\[
  \rho\,\OPT(\Phi)-\gamma_2\eps m
  \le L_{\le\ell_0}^\Theta(\Snap_{\eps}(\Phi_{\le\ell_0}))
     +\rho M_{>\ell_0}
  \le\OPT(\Phi).
\]
\end{lemma}
\end{samepage}
The proof of Lemma~\ref{lem:weighted-snapshot-score-error} is in Appendix~\ref{app:weighted-snapshot-proofs}.
We now compare $\Snap_{\eps}$ with the pseudo-snapshot, using the same noise.
\begin{lemma}[Closeness of the two snapshots]
\label{lem:weighted-snapshot-closeness}
With probability at least $1-O(\eps)$ over $\{f_a\}$ and $g$, we have
\[
  \norm{\PsSnap(\Phi_{\le\ell_0})-\Snap_{\eps}(\Phi_{\le\ell_0})}_1
  \le O(\eps m).
\]
\end{lemma}
The proof of Lemma~\ref{lem:weighted-snapshot-closeness} is in Appendix~\ref{app:weighted-snapshot-proofs}. We combine Lemmas~\ref{lem:weighted-snapshot-score-error} and~\ref{lem:weighted-snapshot-closeness} to prove Theorem~\ref{thm:weighted-pseudosnapshot-to-value}.

\begin{samepage}
\begin{proof}[Proof of Theorem~\ref{thm:weighted-pseudosnapshot-to-value}]
Assume the distance bound in Lemma~\ref{lem:weighted-snapshot-closeness} holds.
Since the score is linear with fixed coefficients and the clause counts $M_\ell$ are unchanged, there is a constant $\gamma_3>0$ such that
\begin{equation}
  \left|L_{\le\ell_0}^\Theta(\PsSnap(\Phi_{\le\ell_0}))
       -L_{\le\ell_0}^\Theta(\Snap_{\eps}(\Phi_{\le\ell_0}))\right|
  \le\gamma_3\eps m.
  \label{eq:weighted-pseudosnapshot-score-distance}
\end{equation}
Lemma~\ref{lem:weighted-snapshot-score-error} therefore gives
\[
  \rho\,\OPT(\Phi)-(\gamma_2+\gamma_3)\eps m
  \le L_{\le\ell_0}^\Theta(\PsSnap(\Phi_{\le\ell_0}))+\rho M_{>\ell_0}
  \le\OPT(\Phi)+\gamma_3\eps m.
\]
Subtracting $\gamma_3\eps m$ proves the claim with $\gamma_0=\gamma_3$ and $\gamma_1=\gamma_2+2\gamma_3$.
\end{proof}
\end{samepage}

\section{Quantum algorithm for weighted pseudo-snapshot estimation}

In this section, we estimate the weighted pseudo-snapshot score for the fixed rounding profile $\Theta=(\ell_0,\bw,\mathcal I,\br)$ and coefficients selected in Lemma~\ref{lem:maxksat-weighted-snapshot-bound}.
Our one-endpoint and two-endpoint estimators use similar constructions  of  the quantum sketch of Kallaugher, Parekh, and Voronova~\cite{KPV23} on streaming Max-DiCut.

\begin{lemma}[Quantum estimation]
\label{lem:maxksat-pairwise-estimator}
For any fixed $\eps\in(0,1/100]$, there is a one-pass quantum streaming algorithm using $O_\eps\!\left(\log^5(n+m)\right)$ qubits and $O_\eps\!\left(\log^5(n+m)\right)$ classical bits whose output $\widehat L_{\le\ell_0}$ satisfies
\[
  \left|\widehat L_{\le\ell_0}
    -L_{\le\ell_0}^\Theta(\PsSnap(\Phi_{\le\ell_0}))\right|
  \le\eps m
\]
with probability at least $15/16$.
\end{lemma}
We first describe the quantum registers and the algorithm for each coordinate;  then combine these coordinate-wise procedures into the full quantum algorithm and prove Lemma~\ref{lem:maxksat-pairwise-estimator} in Section~\ref{sec:proof_of_quantum_algorithm}.

\medskip
At first, we recall and introduce some basic definitions and notations for the algorithm design. 
We use the integer weights, literal copies, and degree intervals $(D_a,D_{a+1}]$ from Section~\ref{sec:weighted-pseudosnapshot-closeness}. Recall $D_0=0$, $D_a=\lfloor(1+\eps^3)^a\rfloor$ for $1\le a\le A$, and $\kappa=\lceil K\eps^{-7}\rceil$ for a sufficiently large constant $K>0$. Thus $D_a=O_\eps(m)$ and $A=O_\eps(\log(m+1))$.
For $2D_a\ge\kappa$, the sampling bits $f_a(C,r,h)$ independently select positive literal copies with probability $\kappa/(2D_a)$; $\{f_a\}$ denotes this family over high-degree intervals. Independently of $\{f_a\}$, the noise $g(v)\sim\Unif[-\eps,\eps]$ is shared by all clauses containing $x_v$.
The pseudo-biases and labels are from Definitions~\ref{def:weighted-pseudobias} and~\ref{def:weighted-pseudosnapshot}.
Throughout this section, constants in $O_\eps(\cdot)$ may also depend on the fixed rounding profile and coefficients.
The estimators ignore clauses with length greater than $\ell_0$.

For $C\in\Phi_{\le\ell_0}$ and $r\in[|C|]$, let $a(C,r)\in\{0,\ldots,A-1\}$ be the index of the degree interval containing its prefix total degree:
\[
  D_{a(C,r)}<d_{v_r(C)}^{\le C}\le D_{a(C,r)+1}.
\]
For $C=C_\nu$, write $\ell_i:=|C_i|$ for $i\in[\nu]$. When $2D_a\ge\kappa$, define the sampled positive-prefix count for interval $a$ by
\[
  s_{v,a}^{\le C}:=
  \sum_{\substack{i\in[\nu]:\ \ell_i\le\ell_0}}
  \sum_{\substack{r\in[\ell_i],\ h\in[w_{\ell_i}]\\ v_r(C_i)=v,\ s_r(C_i)=+}}
  f_a(C_i,r,h).
\]
When $a=a(C,r)$ and $v=v_r(C)$, we have $s_{v,a}^{\le C}=s_v^{\le C}$.

Let $K_{\mathrm{base}}\ge3$ be a constant; we will choose its value in the proof of Lemma~\ref{lem:maxksat-pairwise-estimator}.
For every degree interval, independently choose a scaling factor 
\[
  B_a\sim\Unif\bigl\{\lceil K_{\mathrm{base}}D_{a+1}\rceil,\ldots,
    \lceil2K_{\mathrm{base}}D_{a+1}\rceil\bigr\},
\]
independently of $\{f_a\}$ and $g$, so that $B_a>D_{a+1}+1$.
Write $\omega=(\{f_a\},g,\{B_a\})$ for the public randomness shared by all estimators; their measurement outcomes are independent conditional on $\omega$.

We split the pseudo-snapshot coordinates according to these interval indices.
For degree interval indices $a,b\in\{0,\ldots,A-1\}$, labels $\alpha,\beta\in\sgnset\times[L]$, and $1\le j\le\ell\le\ell_0$, define
\[
  U_{\ell,j,\alpha}^{(a)}
  :=\#\{C\in\Phi_{\le\ell_0}:|C|=\ell,\ 
    a(C,j)=a,
    \widetilde\alpha_j(C)=\alpha\}.
\]
For $1\le j<t\le\ell\le\ell_0$, define
\[
  P_{\ell,j,t,\alpha,\beta}^{(a,b)}
  :=\#\Big\{C\in\Phi_{\le\ell_0}:
    {|C|=\ell,\ a(C,j)=a,
      a(C,t)=b,
      (\widetilde\alpha_j(C),\widetilde\alpha_t(C))=(\alpha,\beta)}\Big\}.
\]
Summing over the interval indices gives the coordinates in Definition~\ref{def:weighted-pseudosnapshot}:
\[
  \widetilde U_{\ell,j,\alpha}=\sum_a U_{\ell,j,\alpha}^{(a)},\qquad
  \widetilde P_{\ell,j,t,\alpha,\beta}=\sum_{a,b}P_{\ell,j,t,\alpha,\beta}^{(a,b)}.
\]

Definition~\ref{def:snapshot-score} gives
\begin{equation}
\label{eq:maxksat-pseudosnap-score-decomp}
\begin{aligned}
  L_{\le\ell_0}^\Theta(\PsSnap(\Phi_{\le\ell_0}))
  ={}&\sum_a\sum_\alpha q_\alpha U_{1,1,\alpha}^{(a)}
    +\sum_{a,b}\sum_{\alpha,\beta}
      \bigl(1-(1-q_\alpha)(1-q_\beta)\bigr)P_{2,1,2,\alpha,\beta}^{(a,b)}\\
  &+\sum_{\ell=3}^{\ell_0}\bigg[
      c_\ell M_\ell
      +\sum_{j=1}^\ell\sum_a\sum_\alpha
        u_{\ell,j}(\alpha)U_{\ell,j,\alpha}^{(a)}\\
  &\hspace{5em}+\sum_{1\le j<t\le\ell}\sum_{a,b}\sum_{\alpha,\beta}
        p_{\ell,jt}(\alpha,\beta)P_{\ell,j,t,\alpha,\beta}^{(a,b)}\bigg].
\end{aligned}
\end{equation}
The counts $M_\ell$, $3\le\ell\le\ell_0$, can be maintained exactly with $O(\log m)$ classical bits, since $\ell_0$ is fixed.

We run the estimators of one-endpoint and two-endpoint coordinate in Sections~\ref{subsec:One-endpoint_estimation} and~\ref{subsec:Two-endpoint_estimation}, respectively. 
We first describe their quantum registers and operations in the following section.

\subsection{Quantum register}

Let $\mathcal U$ be the set of basis labels other than $\bot$, distinct from the literal labels $\alpha\in\sgnset\times[L]$.
A basis state is \textit{live} if its amplitude is nonzero; we use its label to refer to the state.
For the set $T\subseteq\mathcal U$ of live labels, write
\[
  \ket{Q(T)}:=\frac{\ket{\bot}+\sum_{x\in T}\ket{x}}{\sqrt{1+|T|}},
\]
where $\ket{\bot}$ is a fixed anchor state whose amplitude is retained after every zero query outcome, including when $T=\emptyset$.
Let $\gamma_\eps>0$ be a sufficiently large constant. In all sketches below, we will use the quantum register with size $$M:=\lceil  \gamma_{\eps} m \rceil.$$ 

We use the following operations.
\begin{enumerate}[leftmargin=1.5em,itemsep=0.3em,topsep=0.3em]
  \item \emph{Permutation.} For a permutation $\pi$ of $\mathcal U$, apply the unitary
  $\ket{x}\mapsto\ket{\pi(x)}$, fixing $\ket{\bot}$.

  \item \emph{Single-query.} For $x\in\mathcal U$, measure with the projectors
  $\ket{x}\!\bra{x}$ and $I-\ket{x}\!\bra{x}$.
  Output $\zeta_x=1$ on the first outcome and stop the quantum updates.
  On the second outcome, output $0$ and continue with $\ket{Q(T\setminus\{x\})}$.
  Thus
  \begin{equation}\label{eq:single-query-mean}
    \Pr[\zeta_x=1\mid\Qstate{T}]
    =\mathbb E[\zeta_x\mid\Qstate{T}]
    =\frac{\1[x\in T]}{1+|T|}.
  \end{equation}

  \item \emph{Pair-query.} For distinct $x,y\in\mathcal U$, set
  \[
    \ket{\psi^{\pm}_{x,y}}:=\frac{\ket{x}\pm\ket{y}}{\sqrt2}.
  \]
  Measure with the projectors onto these two states and their orthogonal complement.
  Output $\chi_{x,y}=\pm1$ on the corresponding outcome and stop the quantum updates.
  On outcome $0$, continue with $\ket{Q(T\setminus\{x,y\})}$.
  The probabilities and expectation are
  \[
    \Pr[\chi_{x,y}=\pm1\mid\Qstate{T}]
    =\frac{(\1[x\in T]\pm\1[y\in T])^2}{2(1+|T|)},
  \]
  \begin{equation}\label{eq:pair-query-mean}
    \mathbb E[\chi_{x,y}\mid\Qstate{T}]
    =\frac{2\1[x\in T]\1[y\in T]}{1+|T|}.
  \end{equation}
\end{enumerate}

\subsection{One-endpoint estimation} \label{subsec:One-endpoint_estimation}

Fix a one-endpoint coordinate $U^{(a)}_{\ell,j,\alpha}$ with label $\alpha=(s_\alpha,i_\alpha)$. It is \textit{high-degree} when $D_a\ge\kappa/2$ and \textit{low-degree} otherwise. We handle these cases in Sections~\ref{sec:highdegree} and~\ref{Sec:lowdegree}, respectively.

We test the length and literal sign as each clause arrives and use quantum queries to test its prefix counts. A successful query selects one clause; the sub-sketch then counts its exact suffix degrees and uses them with the candidate prefix values to test the bucket. Its output is scaled by $M$ or $M/2$ as specified below.

\paragraph{Prefix counts.}
Recall that $B_a>D_{a+1}+1$. For nonnegative integers $d,c$, define
\[
  \Gamma_a(d,c):=d+B_a c.
\]
In a high-degree sketch, $c$ is the sampled positive-prefix count $s_{v,a}^{\le C}$; in a low-degree sketch, it is the exact positive-prefix degree $p_v^{\le C}$.

For every $q\in\{0,\ldots,\kappa-1\}$, integer $c\ge0$, and degree $D_a<d\le D_{a+1}$,
\begin{equation}\label{eq:one-endpoint-ordinary-high-id}
\1[\Gamma_a(d,c)\ge B_a q+D_a+1]
-\1[\Gamma_a(d,c)\ge B_a q+D_{a+1}+1]
=
\1[c=q].
\end{equation}
Likewise, for integers $D_a<d,d_0\le D_{a+1}$ and $c,c_0\ge0$,
\begin{equation}\label{eq:one-endpoint-ordinary-low-id}
\1[\Gamma_a(d,c)\ge B_a c_0+d_0]
-\1[\Gamma_a(d,c)\ge B_a c_0+d_0+1]
=
\1[d=d_0,\ c=c_0].
\end{equation}
For every integer $q\ge1$, the same bound on $B_a$ gives
\begin{equation}\label{eq:one-endpoint-cumulative-high-id}
\1[\Gamma_a(d,c)\ge qB_a]=\1[c\ge q]
\qquad(D_a<d\le D_{a+1},\ c\ge0).
\end{equation}

\paragraph{Impersonation.}
Set
\[
H_a:=
\begin{cases}
\kappa, & D_a\ge\kappa/2,\\
D_{a+1}, & D_a<\kappa/2.
\end{cases}
\]
We say that an endpoint with prefix degree $d$ \emph{impersonates degree interval $a$} if $d>D_{a+1}$ and
\begin{equation}\label{eq:scale-a-impersonation-def}
d\in
\bigcup_{h=1}^{H_a}
\bigl[hB_a+D_a+1,\ hB_a+D_{a+1}\bigr].
\end{equation}
Let $\mathsf{Imp}^{a}_{C,r}$ denote this event for endpoint $(C,r)$. The condition includes every nonzero threshold difference from outside the queried degree interval, although an endpoint satisfying it may give a zero difference.

\begin{lemma}[Impersonation]\label{lem:low-degree-singleton-impersonation}
Assume $D_a<\kappa/2$. Fix a queried low-degree level $B_a c_0 +d_0$, where $D_a<d_0\le D_{a+1}$ and $0\le c_0 \le d_0$. Let $d^\ast,c^\ast$ be nonnegative integers. If
\[
\Gamma_a(d^\ast,c^\ast)=B_a c_0 +d_0
\]
and $d^\ast\notin(D_a,D_{a+1}]$, then $d^\ast>D_{a+1}$ and the endpoint with prefix degree $d^\ast$ impersonates degree interval $a$.
\end{lemma}

\begin{proof}
If $d^\ast\le D_a$, equality modulo $B_a$ gives $d^\ast=d_0>D_a$, a contradiction. Thus $d^\ast>D_{a+1}$. Writing $h:=c_0-c^\ast$, we have
\[
d^\ast=hB_a+d_0>D_{a+1},\qquad
1\le h\le c_0\le d_0\le D_{a+1}=H_a.
\]
Since $D_a<d_0\le D_{a+1}$, this gives
$d^\ast\in[hB_a+D_a+1,hB_a+D_{a+1}]$, as required.
\end{proof}

\subsubsection{High-degree case} \label{sec:highdegree}

For each $q\in\{0,1,\ldots,\kappa\}$, we maintain an independent
\textit{$q$-Sketch}. The case $q=\kappa$ accounts for all endpoints with
$s_{v,a}^{\le C}\ge\kappa$ through the capped value
$\min\{s_{v,a}^{\le C},\kappa\}$, so high-degree one-endpoint coordinates need no
separate overflow term.

Each $q$-Sketch consists of two independent sub-sketches, called the \textit{Upper-$q$-Sketch} and the \textit{Lower-$q$-Sketch}. Their outputs are combined to give an estimate for the high-degree coordinate.

The quantum register uses the following basis labels. We use the global register size $M=\lceil \gamma_\eps m\rceil$, with $\gamma_\eps$ chosen large enough for the fixed $\eps>0$.
\begin{enumerate}
    \item The fixed anchor basis state $\ket{\bot}$;
    \item Spare labels: $S_1,\ldots, S_M$;
    \item $H$-labels: $H(v,t)$, where $v\in [n]$, $t\in [M]$;
    \item $E$-labels: $E(v,t)$, where $v\in [n]$, $t\in [M]$;
    \item $H$-garbage-labels: $G_{H}(v,t)$,where $v\in [n]$, $t\in [M]$;
    \item $E$-garbage-labels: $G_{E}(v,t)$,where $v\in [n]$, $t\in [M]$;
\end{enumerate}
Besides, we maintain a pointer $\mathcal{P}$ for the spare labels in a classical register, which costs $O(\log M)=O(\log m)$ bits. Because the basis state is specified by a variable index and a level in $[M]$, this quantum register costs $O(\log (nm))=O(\log n )$ qubits, for $m=\mathrm{poly} (n)$.

\paragraph{What the labels are for.}
The spare labels make each update a permutation and are not queried. Each literal copy on $u$ applies one shift to the $H(u,\cdot)$ labels, which are used to test prefix total degree.

The $E$-labels are used to test thresholds of
\[
\Gamma_a(d_u^{\le C},s_{u,a}^{\le C})
=d_u^{\le C}+B_a s_{u,a}^{\le C}.
\]
Each literal copy applies one shift, and each selected positive copy applies $B_a$ additional shifts. The garbage labels make these updates permutations on a finite set of basis labels.

\begin{definition}[Increment operation]
Given a variable $u\in[n]$ and the current  pointer $\mathcal P$, define $\inc{H,u,1}$ as follows.
\begin{enumerate}
    \item Apply the cyclic permutation of basis labels whose cycle is
    \[
    S_{\mathcal P}\to H(u,1)\to H(u,2)\to \cdots \to H(u,M)
    \to G_H(u,M)\to \cdots \to G_H(u,1)\to S_{\mathcal P}.
    \]
    We call this one elementary shift. Its quantum part is unitary because it is
    a permutation of basis labels.
    \item Move the pointer to the next spare label. If $\mathcal P<M$,
set $\mathcal P\leftarrow \mathcal P+1$. If $\mathcal P=M$, set
$\mathcal P\leftarrow 1$. Thus the pointer runs through
$S_1,S_2,\ldots,S_M$ and then starts again from $S_1$.
\end{enumerate}
For any integer $r\ge1$, $\inc{H,u,r}$ means repeating this elementary shift $r$ times, updating $\mathcal P$ after each repetition, that is, \[
  \mathrm{inc}(H,u,r):=\underbrace{\mathrm{inc}(H,u,1)\circ\cdots\circ\mathrm{inc}(H,u,1)}_{r\text{ times}}.
\]
For $E$-labels, we use the same shift operation on the chain $E(u,\cdot)$; in particular, $\inc{E,u,B_a}$ means applying this one-step shift $B_a$ times, so it advances $\mathcal P$ by $B_a$.
\end{definition}

We describe the Lower-$q$-Sketch for $q\ge1$ first, then give the changes for the other sub-sketches. Set $\mathsf{Tgt}^{\ell,j,\alpha}_{C}:=0$ if $|C|\ne\ell$; otherwise set
\[
\mathsf{Tgt}^{\ell,j,\alpha}_{C}:=\1[s_j(C)=s_\alpha].
\]


\medskip

Now we are ready to describe our algorithm.
\paragraph{Initialization.} The quantum register is initialized as: 
\[
  \Qstate{T_0}
  =
  \frac{|\perp\rangle+|S_1\rangle+\cdots+|S_{M-1}\rangle}{\sqrt M},
\]
where $T_0:=\{S_1,S_2,\dots,S_{M-1}\}$ is the initial live set.
Initialize $\mathcal{P}=1$.

When a clause $C$ arrives, the quantum algorithm performs two types of operations in turn: update and query.

\paragraph{Update.} Decompose the clause $C$ into literal copies and handle them one by one. For each literal copy $(C,r,h)$, set $u:=v_r(C)$ and let $T$ be the live set. Perform the following increment operations in turn. Each operation applies a unitary permutation to the quantum register and updates the pointer as specified above.

\medskip
\textbf{Step-1.} $
      \mathrm{inc}(H,u,1).
    $

\medskip    
\textbf{Step-2.}      $
      \mathrm{inc}(E,u,1).
    $
    
\medskip    
\textbf{Step-3.} If this literal copy $(C,r,h)$ is positive and $f_a(C,r,h)=1$,  perform   $
      \mathrm{inc}(E,u,B_a).
    $ 

\paragraph{Query.} If $\mathsf{Tgt}^{\ell,j,\alpha}_{C}=0$, do nothing and end. Otherwise, write $v:=v_j(C)$ for its variable index. Perform the pair-query $$(H(v,D_{a}+1),  E(v,qB_a ))$$ on the current quantum register. Denote the output of pair-query by $\chi_{C}^{\mathrm L}\in \{-1,0,+1\}$.    

\medskip
\textbf{Branch 1:} if $\chi_{C}^{\mathrm L} = 0$, we say this query \textit{fails}. Then do nothing and end.

\medskip
\textbf{Branch 2:} if  $\chi_{C}^{\mathrm L}\in \{-1,+1\}$, we say this query is \textit{successful}. Then this quantum register performs no further updates or queries. During the rest of the stream, the sub-sketch only keeps the following classical data: \begin{enumerate}
    \item store the variable $v$ and clause $C$;
    \item store the nonzero measurement outcome $\chi_C^{\mathrm L}$;
    \item store $q$;
    \item in the following streaming of clauses, count the suffix total degree  $d_{v}^{>C}$ and suffix positive count $p_{v}^{>C}$ accurately.
\end{enumerate}
The classical part costs $O(\log (n+m))$ bits. We say a sub-sketch is \textit{active} if the quantum part has not ended, equivalently, if it has not entered Branch 2 yet.

\paragraph{Output.} If a high-degree sub-sketch succeeds on $C$, then after the rest of the stream has been processed it has the exact suffix counts $p_v^{>C}$ and $d_v^{>C}$. For every $r\in\{0,\ldots,\kappa\}$, define
\[
F_r(C)
:=
\1 \left[
\clip
\left(
\frac{2\left(\frac{2D_a}{\kappa}r+p_v^{>C}\right)}
     {D_a+d_v^{>C}}
-1+g(v)
\right)
\in I_{i_\alpha}
\right].
\]
Finally,  the Lower-$q$-Sketch $(q\ge 1)$ outputs 
\[
  Y_q^{\mathrm L}(C):=\frac M2\cdot \chi_{C}^{\mathrm L}\cdot (F_q(C)-F_{q-1}(C)). 
\]
In all other cases,  the Lower-$q$-Sketch $(q\ge 1)$ outputs $0$.

\medskip

For the other sub-sketches, only the following operations and outputs change.

\paragraph{Upper-$q$-Sketch ($q\ge 1$).} In the query step, we perform the pair-query on \[
(H(v,D_{a+1}+1), 
E(v,qB_a ))
\]
rather than $(H(v,D_{a}+1), E(v,qB_a ))$ for Lower-$q$-Sketch.

Similarly, define $\chi^{\mathrm U}_{C}\in \{-1,0,1\}$ for the output of the pair-query. In its successful branch, the Upper-$q$-Sketch stores the nonzero value of $\chi_C^{\mathrm U}$ together with the same classical data. For Branch $2$, define the final output as
\[
  Y_q^{\mathrm U}(C):=\frac M2\cdot \chi_{C}^{\mathrm U}\cdot (F_q(C)-F_{q-1}(C)).
\]
In all other cases, it outputs  $0$.

\paragraph{Lower-$0$-Sketch and Upper-$0$-Sketch.} These sub-sketches use the anchor, spare labels, $H$-labels, and $H$-garbage-labels, and apply only $\inc{H,u,1}$.

In the query step, the Lower-$0$-Sketch performs a single-query on $H(v,D_a+1)$ and the Upper-$0$-Sketch performs a single-query on $
H(v,D_{a+1}+1)$.

Define $\xi_{C}^{\mathrm L}$ and $\xi_{C}^{\mathrm U}$ to denote the outputs of the single-query from Lower-$0$-Sketch and  Upper-$0$-Sketch  respectively. In a successful branch the sketch stores this nonzero output, which equals $1$, and the same classical data and suffix counters as above. Define the final outputs as \[
  Y_0^{\mathrm L}(C):=M\cdot \xi_{C}^{\mathrm L}\cdot F_0(C),
  \qquad
  Y_0^{\mathrm U}(C):=M\cdot \xi_{C}^{\mathrm U}\cdot F_0(C),
\]
In all other cases, output $0$.

\medskip
Next, we prove the correctness of the algorithm and show that its output provides a good estimate of the coordinate $U^{(a)}_{\ell,j,\alpha}$.

\paragraph{Correctness.}
Fix a sub-sketch and a label $F(u,t)$ that it queries, where $F$ is $H$ or $E$. For this fixed family and variable, the queried level $t$ is fixed, and at least one ordinary $+1$ update occurs between consecutive queries of $F(u,t)$. Immediately before a query, let $d_u$ count the ordinary literal-copy updates applied to this chain so far and $c_u$ count its additional $B_a$-shift updates. The elementary shifts inside a $B_a$-shift update are not counted again in $d_u$. Set $c_u=0$ for a family with no additional updates.

\begin{samepage}
\begin{lemma}\label{lem:nonwrapping-threshold-invariant}
Suppose $1\le t<M$ and, before each query of $F(u,t)$, the sub-sketch has used fewer than $M$ elementary shifts in total. Then, immediately before each such query,
\[
F(u,t)\in T
\quad\Longleftrightarrow\quad
d_u+B_a c_u\ge t.
\]
\end{lemma}

The proof of Lemma~\ref{lem:nonwrapping-threshold-invariant} is in Appendix~\ref{app:one-endpoint-threshold-proof}.
\end{samepage}

\begin{remark}
Lemma~\ref{lem:nonwrapping-threshold-invariant} also applies to the $Z_1$- and $Z_2$-labels in Section~\ref{subsec:Two-endpoint_estimation}: each chain has a fixed queried level, and each queried clause updates that chain before the query.
\end{remark}

The next lemma bounds both the queried levels and the number of shifts, including the low-degree sketches described below.

\begin{lemma}[No-wrap event for one-endpoint sketches]\label{lem:good_label_budget}
For each high-degree interval $a$, let $X_a$ count all positive literal copies selected by $f_a$ over the whole stream. For a sufficiently large constant $K_\eps$, define
\[
\mathcal G_1:=
\bigcap_{a:\,D_a\ge\kappa/2}
\left\{X_a\le K_\eps\frac{m}{D_a}\right\}.
\]
Then $\Pr[\mathcal G_1]\ge1-1/160$. Choosing $\gamma_\eps$ in $M=\lceil\gamma_\eps m\rceil$ sufficiently large ensures that every queried level satisfies $1\le t<M$ and, on $\mathcal G_1$, every one-endpoint sub-sketch performs fewer than $M$ elementary shifts over the whole stream.
\end{lemma}

The proof of Lemma~\ref{lem:good_label_budget} is in Appendix~\ref{app:one-endpoint-budget-proof}.

The updates remain permutations even if $\mathcal G_1$ fails. On $\mathcal G_1$, Lemma~\ref{lem:nonwrapping-threshold-invariant} identifies every queried label with its prefix-count threshold.

\paragraph{Estimation.}
Fix public randomness $\omega=(\{f_a\},g,\{B_a\})$ satisfying $\mathcal G_1$. All expectations below average over the quantum measurements conditional on $\omega$. The lower and upper sketches select the degree interval; a telescoping sum over $q$ then tests the bucket.

Extend every sub-sketch output by $0$ to clauses with $\mathsf{Tgt}^{\ell,j,\alpha}_{C}=0$. For each clause $C$, define
\begin{equation}\label{eq:estimation}
\widehat U^{(a)}_{\ell,j,\alpha}(C)
:=
Y_0^{\mathrm L}(C)-Y_0^{\mathrm U}(C)
+\sum_{q=1}^{\kappa}\bigl(Y_q^{\mathrm L}(C)-Y_q^{\mathrm U}(C)\bigr),
\end{equation}
and set
\[
\widehat U^{(a)}_{\ell,j,\alpha}
:=\sum_{C\in\Phi_{\le\ell_0}}\widehat U^{(a)}_{\ell,j,\alpha}(C).
\]
Fix a clause $C$ with $\mathsf{Tgt}^{\ell,j,\alpha}_{C}=1$ and write $v:=v_j(C)$. When $a=a(C,j)$, the high-degree prefix approximations give
\[
\widetilde d_v^{\le C}=D_a,\qquad
\widetilde p_v^{\le C}
=\frac{2D_a}{\kappa}\min\{s_{v,a}^{\le C},\kappa\}.
\]
Using $F_q(C)$ defined in the output step, we obtain
\begin{equation}\label{eq:Fq}
\begin{aligned}
\1[\widetilde\alpha_j(C)=\alpha]
&=F_{\min\{s_{v,a}^{\le C},\kappa\}}(C)\\
&=F_0(C)+\sum_{q=1}^{\kappa}
\bigl(F_q(C)-F_{q-1}(C)\bigr)\1[s_{v,a}^{\le C}\ge q].
\end{aligned}
\end{equation}

The following probability cancels the normalization of the active quantum register.

\begin{lemma}[Active probability]\label{lem:active_prob}
Fix a sub-sketch and public randomness $\omega$. Conditioned on this sub-sketch being active just before its query on $C$, let $T$ be its live set and $M'=1+|T|$. Then
\[
\Pr[\text{this sub-sketch is active just before its query on }C\mid\omega]
=\frac{M'}{M}.
\]
\end{lemma}

\begin{proof}
Once $\omega$ is fixed, the updates and previous failed queries determine the live set on the active branch. List the earlier queries as $\text{Query}_1,\ldots,\text{Query}_r$, and let $M_i$ be one plus the number of live labels just before $\text{Query}_i$ on this branch. Set $M_{r+1}=M'$; initially $M_1=M$.

A failed query deletes zero, one, or two live labels. Updates preserve their number, so the conditional probability of remaining active is
\[
\Pr[\text{active after }\text{Query}_i\mid\omega,\ \text{active before }\text{Query}_i]
=1-\frac{M_i-M_{i+1}}{M_i}
=\frac{M_{i+1}}{M_i}.
\]
Multiplying these probabilities gives $M'/M$.
\end{proof}

\paragraph{For $q\in\{1,\ldots,\kappa\}$.}
Conditioned on the Lower-$q$-Sketch being active before $C$, let $T$ be its live set and $M'=1+|T|$. By \eqref{eq:pair-query-mean}, its pair-query on $(H(v,D_a+1),E(v,qB_a))$ gives
\[
\mathbb E[\chi_C^{\mathrm L}\mid\omega,\ \text{Lower-$q$-Sketch active before }C]
=\frac{2}{M'}\1[H(v,D_a+1)\in T]\1[E(v,qB_a)\in T].
\]
Lemma~\ref{lem:nonwrapping-threshold-invariant} replaces the live-label indicators by prefix-count inequalities. Multiplying the query mean by $\frac M2(F_q(C)-F_{q-1}(C))$ and the active probability $M'/M$ from Lemma~\ref{lem:active_prob} gives
\[
\begin{aligned}
\mathbb E[Y_q^{\mathrm L}(C)\mid\omega]
={}&\bigl(F_q(C)-F_{q-1}(C)\bigr)
\1[\mathsf{Tgt}^{\ell,j,\alpha}_{C}=1,\ d_v^{\le C}\ge D_a+1]\\
&\cdot\1[\Gamma_a(d_v^{\le C},s_{v,a}^{\le C})\ge qB_a].
\end{aligned}
\]
Applying the same argument separately to the Upper-$q$-Sketch gives
\[
\begin{aligned}
\mathbb E[Y_q^{\mathrm U}(C)\mid\omega]
={}&\bigl(F_q(C)-F_{q-1}(C)\bigr)
\1[\mathsf{Tgt}^{\ell,j,\alpha}_{C}=1,\ d_v^{\le C}\ge D_{a+1}+1]\\
&\cdot\1[\Gamma_a(d_v^{\le C},s_{v,a}^{\le C})\ge qB_a].
\end{aligned}
\]
Their difference is zero outside $(D_a,D_{a+1}]$. Inside this interval, \eqref{eq:one-endpoint-cumulative-high-id} gives $\Gamma_a(d_v^{\le C},s_{v,a}^{\le C})\ge qB_a\iff s_{v,a}^{\le C}\ge q$, so
\begin{equation}\label{eq:bigq}
\begin{aligned}
\mathbb E[Y_q^{\mathrm L}(C)-Y_q^{\mathrm U}(C)\mid\omega]
={}&\bigl(F_q(C)-F_{q-1}(C)\bigr)\\
&\cdot\1[\mathsf{Tgt}^{\ell,j,\alpha}_{C}=1,\ D_a<d_v^{\le C}\le D_{a+1},\ s_{v,a}^{\le C}\ge q].
\end{aligned}
\end{equation}
Thus high-degree one-endpoint estimation has no impersonation error.

\paragraph{For $q=0$.}
For the Lower-$0$-Sketch with live set $T$ and $M'=1+|T|$, \eqref{eq:single-query-mean} gives
\[
\mathbb E[\xi_C^{\mathrm L}\mid\omega,\ \text{Lower-$0$-Sketch active before }C]
=\frac{\1[H(v,D_a+1)\in T]}{M'}.
\]
Multiplying by $MF_0(C)$ and its active probability $M'/M$ from Lemma~\ref{lem:active_prob}, and applying the same argument separately to the Upper-$0$-Sketch, gives
\begin{equation}\label{0q}
\mathbb E[Y_0^{\mathrm L}(C)-Y_0^{\mathrm U}(C)\mid\omega]
=F_0(C)\1[\mathsf{Tgt}^{\ell,j,\alpha}_{C}=1,\ D_a<d_v^{\le C}\le D_{a+1}].
\end{equation}
Summing the conditional expectations over $q$ and using the telescoping identity~\eqref{eq:Fq} gives
\begin{equation}\label{eq:one-endpoint-high-coordinate}
\mathbb E[\widehat U^{(a)}_{\ell,j,\alpha}(C)\mid\omega]
=\1[\mathsf{Tgt}^{\ell,j,\alpha}_{C}=1,\ D_a<d_v^{\le C}\le D_{a+1},\ \widetilde\alpha_j(C)=\alpha].
\end{equation}
Summing over clauses yields
\[
\mathbb E[\widehat U^{(a)}_{\ell,j,\alpha}\mid\omega]
=U^{(a)}_{\ell,j,\alpha}.
\]

\subsubsection{Low-degree case} \label{Sec:lowdegree}

For a low-degree coordinate $U^{(a)}_{\ell,j,\alpha}$, there are $O_\eps(1)$ candidate prefix-degree pairs $(d,c)$, with integers $D_a<d\le D_{a+1}$ and $0\le c\le d$. For each pair, we maintain a $(d,c)$-Sketch.

Each $(d,c)$-Sketch consists of independent Lower-$(d,c)$-Sketch and Upper-$(d,c)$-Sketch sub-sketches. Their quantum registers have the same basis labels. The quantum register and update operations are the same as in the high-degree case, except for the positive-count update described next.

For the update step, Step-3 is replaced as follows: if this literal copy is positive, perform
$
      \mathrm{inc}(E,u,B_a).
    $

For the query step, the Lower-$(d,c)$-Sketch performs a single-query on \[
E(v, B_a c+d )
\] and the Upper-$(d,c)$-Sketch performs single-query on \[
E(v, B_a c+d+1 ).
\] 

The difference of their expected outputs tests $\Gamma_a(d_v^{\le C},p_v^{\le C})=B_ac+d$. Since it does not separately test the prefix degree interval, an endpoint with degree outside $(D_a,D_{a+1}]$ may contribute. Lemma~\ref{lem:low-degree-singleton-impersonation} includes every such contribution in $\mathsf{Imp}^{a}_{C,j}$.

\paragraph{Output.}
After a successful query on $C$, the corresponding Lower-$(d,c)$-Sketch or Upper-$(d,c)$-Sketch stores the nonzero single-query output, which equals $1$, and computes
\[
F_{d,c}(C)
:=
\1\!\left[
\clip\!\left(
\frac{2(c+p_v^{>C})}{d+d_v^{>C}}-1+g(v)
\right)\in I_{i_\alpha}
\right]
\]
using the candidate prefix values $(d,c)$ in Definition~\ref{def:weighted-pseudobias}. Define
\[
Y_{d,c}^{\mathrm L}(C):=M\cdot \xi^{\mathrm L}_{d,c}(C)\cdot F_{d,c}(C),
\qquad
Y_{d,c}^{\mathrm U}(C):=M\cdot \xi^{\mathrm U}_{d,c}(C)\cdot F_{d,c}(C),
\]
where $\xi^{\mathrm L}_{d,c}(C),\xi^{\mathrm U}_{d,c}(C)\in\{0,1\}$ are the single-query outputs of the Lower-$(d,c)$-Sketch and Upper-$(d,c)$-Sketch, respectively. Extend both outputs by $0$ to clauses with $\mathsf{Tgt}^{\ell,j,\alpha}_{C}=0$.

\paragraph{Estimation and Analysis.}
Fix $\omega$ satisfying $\mathcal G_1$. For each clause $C$, define
\[
\widehat U^{(a)}_{\ell,j,\alpha}(C)
:=\sum_{d=D_a+1}^{D_{a+1}}\sum_{c=0}^{d}
\bigl(Y_{d,c}^{\mathrm L}(C)-Y_{d,c}^{\mathrm U}(C)\bigr),
\]
and set
\[
\widehat U^{(a)}_{\ell,j,\alpha}
:=\sum_{C\in\Phi_{\le\ell_0}}\widehat U^{(a)}_{\ell,j,\alpha}(C).
\]
Fix a clause $C$ with $\mathsf{Tgt}^{\ell,j,\alpha}_{C}=1$ and write $v:=v_j(C)$. Conditioned on the Lower-$(d,c)$-Sketch being active before $C$, let its live set be $T$ and $M'=1+|T|$. Equation~\eqref{eq:single-query-mean} gives the single-query mean; replacing its live-label indicator by the prefix-count inequality from Lemma~\ref{lem:nonwrapping-threshold-invariant} gives
\[
\begin{aligned}
&\mathbb E[Y_{d,c}^{\mathrm L}(C)\mid\omega,\ \text{Lower-$(d,c)$-Sketch active before }C]\\
&\qquad=\frac{M}{M'}F_{d,c}(C)
\1[\Gamma_a(d_v^{\le C},p_v^{\le C})\ge B_ac+d].
\end{aligned}
\]
Multiplying by its active probability $M'/M$, and applying the same argument separately to the Upper-$(d,c)$-Sketch, yields
\[
\mathbb E[Y_{d,c}^{\mathrm L}(C)-Y_{d,c}^{\mathrm U}(C)\mid\omega]
=F_{d,c}(C)\1[\Gamma_a(d_v^{\le C},p_v^{\le C})=B_ac+d].
\]
If $\mathsf{Imp}^{a}_{C,j}$ does not occur, Lemma~\ref{lem:low-degree-singleton-impersonation} excludes equality outside $(D_a,D_{a+1}]$, while \eqref{eq:one-endpoint-ordinary-low-id} identifies the exact prefix-degree pair inside this interval. Hence
\[
\Gamma_a(d_v^{\le C},p_v^{\le C})=B_ac+d
\quad\Longleftrightarrow\quad
d_v^{\le C}=d,\quad p_v^{\le C}=c.
\]
For this pair, $(d,c)$ equals the prefix approximations, so it gives the same pseudo-bias; with sign $s_\alpha$, $F_{d,c}(C)=\1[\widetilde\alpha_j(C)=\alpha]$ by Definitions~\ref{def:weighted-pseudobias} and~\ref{def:weighted-pseudosnapshot}. Therefore
\[
\mathbb E[\widehat U^{(a)}_{\ell,j,\alpha}(C)\mid\omega]
=\1[D_a<d_v^{\le C}\le D_{a+1},\ \widetilde\alpha_j(C)=\alpha]
\]
whenever $\mathsf{Imp}^{a}_{C,j}$ does not occur.

For any fixed clause, at most one candidate pair satisfies $\Gamma_a(d_v^{\le C},p_v^{\le C})=B_ac+d$: equality for two pairs implies $d-d'=B_a(c'-c)$, and $|d-d'|<B_a$ forces $c=c'$ and $d=d'$. Thus the expected contribution of all $(d,c)$-sketches lies in $[0,1]$. An impersonating endpoint has $d_v^{\le C}>D_{a+1}$ and is not counted by the target coordinate. Clauses with $\mathsf{Tgt}^{\ell,j,\alpha}_{C}=0$ have zero output, so
\begin{equation}\label{eq:one-endpoint-low-coordinate}
\begin{aligned}
\left|
\mathbb E[\widehat U^{(a)}_{\ell,j,\alpha}(C)\mid\omega]
-\1\!\left[\substack{\mathsf{Tgt}^{\ell,j,\alpha}_{C}=1,\\
D_a<d_v^{\le C}\le D_{a+1},\
\widetilde\alpha_j(C)=\alpha}\right]
\right|
\le\1[\mathsf{Tgt}^{\ell,j,\alpha}_{C}=1,\ \mathsf{Imp}^{a}_{C,j}].
\end{aligned}
\end{equation}
Summing over clauses gives
\begin{equation}\label{eq:one-endpoint-low-coordinate-total}
\left|
\mathbb E[\widehat U^{(a)}_{\ell,j,\alpha}\mid\omega]
-U^{(a)}_{\ell,j,\alpha}
\right|
\le\sum_{C\in\Phi_{\le\ell_0}}
\1[\mathsf{Tgt}^{\ell,j,\alpha}_{C}=1,\ \mathsf{Imp}^{a}_{C,j}].
\end{equation}

\subsection{Two-endpoint estimation} \label{subsec:Two-endpoint_estimation}
Fix a two-endpoint coordinate $P^{(a,b)}_{\ell,j,t,\alpha,\beta}$, where $1\le j<t\le\ell\le\ell_0$, $\alpha=(s_\alpha,i_\alpha)$, and $\beta=(s_\beta,i_\beta)$.
The indices $a,b$ specify the degree intervals $(D_a,D_{a+1}]$ and $(D_b,D_{b+1}]$.
For $|C|=\ell$, set $\mathsf{Tgt}^{\ell,j,t,\alpha,\beta}_C:=\1[s_j(C)=s_\alpha,\ s_t(C)=s_\beta]$; set it to $0$ otherwise.
When $\mathsf{Tgt}^{\ell,j,t,\alpha,\beta}_C=1$, write $u:=v_j(C)$ and $v:=v_t(C)$.
We use the sampled prefix counts $s_{u,a}^{\le C}$ and $s_{v,b}^{\le C}$ on high-degree intervals, and the exact positive prefix degrees on low-degree intervals.

For fixed $\omega$, the expected output can differ from the coordinate if an endpoint satisfies the impersonation condition in~\eqref{eq:scale-a-impersonation-def} for a different degree interval, or if a sampled prefix count on its actual interval is at least $\kappa$.
The final bound below accounts for these two errors.

\subsubsection{High-degree on both endpoints} \label{sec:highhigh}
First suppose $2D_a\ge\kappa$ and $2D_b\ge\kappa$.
Recall
\[
  \Gamma_a(d,c)=d+B_a c,\qquad \Gamma_b(d,c)=d+B_b c,
\]
where $B_a,B_b$ are chosen at the start of this section.
Set
\[
  \Sigma_a:=\{0,1,\ldots,\kappa-1\},\qquad
  \Sigma_b:=\{0,1,\ldots,\kappa-1\}.
\]
For each $(\sigma,\tau)\in\Sigma_a\times\Sigma_b$, maintain four sub-sketches with separate quantum registers: Lower-Lower, Lower-Upper, Upper-Lower, and Upper-Upper.
They use the same $\{f_a\},g,\{B_a\}$; their quantum measurements are independent conditioned on $\omega$.
The superscripts $\mathrm{LL},\mathrm{LU},\mathrm{UL},\mathrm{UU}$ refer to these four sub-sketches, respectively.

\paragraph{Quantum register.}
Each of the four sub-sketches has an identical quantum register. We use the global register size $M=\lceil \gamma_{\eps} m\rceil$, with $\gamma_\eps$ chosen large enough for the fixed $\eps>0$. The basis labels are:
\begin{enumerate}
  \item the anchor $\ket{\bot}$;
  \item spare labels $S_1,\ldots,S_M$;
  \item labels $Z_1(v,r)$ and $Z_2(v,r)$, where $v\in[n]$ and $r\in[M]$;
  \item garbage labels $G_1(v,r)$ and $G_2(v,r)$, where $v\in[n]$ and $r\in[M]$.
\end{enumerate}
Each sub-sketch has a classical pointer $\mathcal P$, shared by its two label families.
For $i=1,2$, $\inc{Z_i,u,r}$ repeats the following cyclic permutation $r$ times, advancing $\mathcal P$ after each repetition:
\[
  S_{\mathcal P}\to Z_i(u,1)\to\cdots\to Z_i(u,M)
  \to G_i(u,M)\to\cdots\to G_i(u,1)\to S_{\mathcal P}.
\]
After $S_M$, the pointer returns to $S_1$.
Initialize
\[
  \Qstate{T_0}
  =\frac{\ket{\bot}+\ket{S_1}+\cdots+\ket{S_{M-1}}}{\sqrt M},
  \qquad T_0:=\{S_1,\ldots,S_{M-1}\},\qquad\mathcal P:=1.
\]
The register uses $O(\log(n+m))$ qubits, and the pointer uses $O(\log m)$ classical bits.

\paragraph{Update.}
For each literal copy on a variable $u$, each active sub-sketch performs
\begin{enumerate}
  \item $\inc{Z_1,u,1}$;
  \item $\inc{Z_2,u,1}$;
  \item if the copy is positive and selected by $f_a$, $\inc{Z_1,u,B_a}$;
  \item if the copy is positive and selected by $f_b$, $\inc{Z_2,u,B_b}$.
\end{enumerate}
The total numbers of shifts on $Z_1(u,\cdot)$ and $Z_2(u,\cdot)$ are $\Gamma_a(d_u^{\le C},s_{u,a}^{\le C})$ and $\Gamma_b(d_u^{\le C},s_{u,b}^{\le C})$, respectively, before a query.

\paragraph{Query.}
If $\mathsf{Tgt}^{\ell,j,t,\alpha,\beta}_C=0$, do nothing.
Otherwise the four sub-sketches apply pair-query to the following pairs, in the order $\mathrm{LL},\mathrm{LU},\mathrm{UL},\mathrm{UU}$:
\[
\begin{aligned}
  &\bigl(Z_1(u,B_a\sigma+D_a+1),\ Z_2(v,B_b\tau+D_b+1)\bigr),\\
  &\bigl(Z_1(u,B_a\sigma+D_a+1),\ Z_2(v,B_b\tau+D_{b+1}+1)\bigr),\\
  &\bigl(Z_1(u,B_a\sigma+D_{a+1}+1),\ Z_2(v,B_b\tau+D_b+1)\bigr),\\
  &\bigl(Z_1(u,B_a\sigma+D_{a+1}+1),\ Z_2(v,B_b\tau+D_{b+1}+1)\bigr).
\end{aligned}
\]
Denote their outputs by $\chi^{\mathrm{LL}}_{\sigma,\tau}(C)$, $\chi^{\mathrm{LU}}_{\sigma,\tau}(C)$, $\chi^{\mathrm{UL}}_{\sigma,\tau}(C)$, and $\chi^{\mathrm{UU}}_{\sigma,\tau}(C)$, each in $\{-1,0,1\}$.
A sub-sketch continues when its output is $0$.
When its output is $\pm1$, it stops the quantum updates and queries, stores the sign, $u,v,\sigma,\tau$, and counts
\[
  d_u^{>C},\quad p_u^{>C},\quad d_v^{>C},\quad p_v^{>C}
\]
exactly during the remaining stream.

\paragraph{Output.}
For a sub-sketch that succeeds on $C$, test
\[
\begin{aligned}
  F_{\sigma,\tau}(C):=\1\Biggl[&
  \clip\left(\frac{2\left(\frac{2D_a}{\kappa}\sigma+p_u^{>C}\right)}{D_a+d_u^{>C}}-1+g(u)\right)\in I_{i_\alpha},\\
  &\clip\left(\frac{2\left(\frac{2D_b}{\kappa}\tau+p_v^{>C}\right)}{D_b+d_v^{>C}}-1+g(v)\right)\in I_{i_\beta}\Biggr].
\end{aligned}
\]
For $h\in\{\mathrm{LL},\mathrm{LU},\mathrm{UL},\mathrm{UU}\}$, define its contribution on that clause by
\[
  Y^h_{\sigma,\tau}(C):=\frac M2\chi^h_{\sigma,\tau}(C)F_{\sigma,\tau}(C).
\]
Its contribution on every other clause is $0$.

\subsubsection{Remaining cases} \label{sec:othercases}
If either interval is low-degree, keep the same four sub-sketches and registers.
On a low-degree interval use the exact positive prefix degree instead of the sampled count.
Set
\[
  \Sigma_a:=\begin{cases}
    \{0,1,\ldots,\kappa-1\},&2D_a\ge\kappa,\\
    \{(d_1,c_1):D_a<d_1\le D_{a+1},\ 0\le c_1\le d_1\},&2D_a<\kappa,
  \end{cases}
\]
\[
  \Sigma_b:=\begin{cases}
    \{0,1,\ldots,\kappa-1\},&2D_b\ge\kappa,\\
    \{(d_2,c_2):D_b<d_2\le D_{b+1},\ 0\le c_2\le d_2\},&2D_b<\kappa.
  \end{cases}
\]
All degrees and counts in these sets are integers.

\paragraph{Update.}
For every literal copy on $u$, apply $\inc{Z_1,u,1}$ and $\inc{Z_2,u,1}$.
For each positive copy, also apply $\inc{Z_1,u,B_a}$ when $2D_a<\kappa$; when $2D_a\ge\kappa$, apply it only if the copy is selected by $f_a$.
Use the corresponding rule with $b,f_b$ for $Z_2$.

\paragraph{Query.}
When $\mathsf{Tgt}^{\ell,j,t,\alpha,\beta}_C=1$, the lower and upper queried labels for the first family are
\[
\begin{cases}
  Z_1(u,B_a\sigma+D_a+1),\quad Z_1(u,B_a\sigma+D_{a+1}+1),&2D_a\ge\kappa,\\
  Z_1(u,B_a c_1+d_1),\quad Z_1(u,B_a c_1+d_1+1),&2D_a<\kappa,
\end{cases}
\]
where $\sigma=(d_1,c_1)$ in the second case.
The second family uses the same formulas with $b,v,\tau=(d_2,c_2)$ in place of $a,u,\sigma=(d_1,c_1)$.
Pair-query combines the lower/lower, lower/upper, upper/lower, and upper/upper queried labels.
Successful sub-sketches store the same data and count the exact suffix degrees as above.

\paragraph{Output.}
Use the following candidate prefix values:
\[
  (\widehat d_u,\widehat p_u):=\begin{cases}
    (D_a,\frac{2D_a}{\kappa}\sigma),&2D_a\ge\kappa,\\
    (d_1,c_1),&2D_a<\kappa,\quad\sigma=(d_1,c_1),
  \end{cases}
\]
\[
  (\widehat d_v,\widehat p_v):=\begin{cases}
    (D_b,\frac{2D_b}{\kappa}\tau),&2D_b\ge\kappa,\\
    (d_2,c_2),&2D_b<\kappa,\quad\tau=(d_2,c_2).
  \end{cases}
\]
Then set
\[
\begin{aligned}
  F_{\sigma,\tau}(C):=\1\Biggl[&
  \clip\left(\frac{2(\widehat p_u+p_u^{>C})}{\widehat d_u+d_u^{>C}}-1+g(u)\right)\in I_{i_\alpha},\\
  &\clip\left(\frac{2(\widehat p_v+p_v^{>C})}{\widehat d_v+d_v^{>C}}-1+g(v)\right)\in I_{i_\beta}\Biggr].
\end{aligned}
\]
Use the same outputs $Y^h_{\sigma,\tau}(C)$ as in the high-degree case.

\paragraph{Correctness.}
We first bound the number of shifts and queried levels, so Lemma~\ref{lem:nonwrapping-threshold-invariant} applies to both label families.
\begin{lemma}[Increment operations and queried levels]
\label{lem:two-endpoint-label-budget}
Let $\mathcal G_2:=\mathcal G_1$ be the event in Lemma~\ref{lem:good_label_budget}.
For $m\ge1$ and a sufficiently large constant $\gamma_\eps$ in the register size $M=\lceil\gamma_\eps m\rceil$, every queried level lies in $[M-1]$.
On $\mathcal G_2$, every two-endpoint sub-sketch performs fewer than $M$ elementary shifts, counting both label families.
In particular, $\Pr[\mathcal G_2]\ge1-1/160$.
\end{lemma}
\begin{proof}
Recall from Lemma~\ref{lem:good_label_budget} that, on $\mathcal G_1$, the number $X_a$ of positive literal copies selected by $f_a$ over the whole stream satisfies
\[
  X_a\le K_\eps\frac{m}{D_a}\qquad(2D_a\ge\kappa).
\]
Since $B_a\le2K_{\mathrm{base}}D_{a+1}+1$ and $D_{a+1}\le(1+\eps^3)D_a+2$,
\[
  B_aX_a\le(2K_{\mathrm{base}}D_{a+1}+1)K_\eps\frac{m}{D_a}=O_\eps(m).
\]
For $2D_a<\kappa$, every positive copy applies the $B_a$-step update, and
\[
  B_a\sum_v\widetilde p_v=O_\eps(m),
\]
since $B_a=O_\eps(1)$ and~\eqref{eq:weighted-total-copies} gives
\[
\sum_v\widetilde p_v\le\sum_v\widetilde d_v
\le m\max_{\ell\in[\ell_0]}\ell w_\ell=O(m).
\]
The same bounds hold for the second family.
When both intervals are high-degree, the total number of elementary shifts is at most
\[
  2\sum_v\widetilde d_v+B_aX_a+B_bX_b=O_\eps(m)<M.
\]
For a low-degree interval, replace its term $B_aX_a$ or $B_bX_b$ by $B_a\sum_v\widetilde p_v$ or $B_b\sum_v\widetilde p_v$; the same bound holds.

For a high-degree interval, every queried level is at most
\[
  (\kappa-1)B_a+D_{a+1}+1=O_\eps(m).
\]
For a low-degree interval, every queried level is at most
\[
  B_aD_{a+1}+D_{a+1}+1=O_\eps(1).
\]
Choose $\gamma_\eps$ large enough that these levels and the total number of shifts are strictly below $M$.
The probability bound follows from $\mathcal G_2=\mathcal G_1$.
\end{proof}

\paragraph{Threshold inequalities.}
Fix $\omega\in\mathcal G_2$.
For a clause $C$ with $\mathsf{Tgt}^{\ell,j,t,\alpha,\beta}_C=1$, define
\[
  N_a(C):=\begin{cases}
    \Gamma_a(d_u^{\le C},s_{u,a}^{\le C}),&2D_a\ge\kappa,\\
    \Gamma_a(d_u^{\le C},p_u^{\le C}),&2D_a<\kappa.
  \end{cases}
\]
The lower and upper thresholds are
\[
  (L_a(\sigma),U_a(\sigma)):=\begin{cases}
    (B_a\sigma+D_a+1,\ B_a\sigma+D_{a+1}+1),&2D_a\ge\kappa,\\
    (B_a c_1+d_1,\ B_a c_1+d_1+1),&2D_a<\kappa,\quad\sigma=(d_1,c_1).
  \end{cases}
\]
Define $N_b(C),L_b(\tau),U_b(\tau)$ by the same formulas for the second endpoint.
For each sub-sketch, let $t_a$ and $t_b$ be its queried thresholds on the two families.
Lemma~\ref{lem:nonwrapping-threshold-invariant} gives, before its query while active,
\[
  Z_1(u,t_a)\in T\iff N_a(C)\ge t_a,\qquad
  Z_2(v,t_b)\in T\iff N_b(C)\ge t_b.
\]
Put
\[
  \Delta_a(C,\sigma):=\1[N_a(C)\ge L_a(\sigma)]-\1[N_a(C)\ge U_a(\sigma)].
\]
If the first endpoint does not satisfy $\mathsf{Imp}^{a}_{C,j}$, then
\[
  \Delta_a(C,\sigma)=\begin{cases}
    \1[D_a<d_u^{\le C}\le D_{a+1},\ s_{u,a}^{\le C}=\sigma],&2D_a\ge\kappa,\\
    \1[d_u^{\le C}=d_1,\ p_u^{\le C}=c_1],&2D_a<\kappa,\quad\sigma=(d_1,c_1).
  \end{cases}
\]
On its actual degree interval, the threshold differences in~\eqref{eq:one-endpoint-ordinary-high-id} and~\eqref{eq:one-endpoint-ordinary-low-id} test the sampled count and the exact prefix-degree pair, respectively.
If the prefix degree is outside the interval and $\Delta_a(C,\sigma)\ne0$, then for some queried count $c_0$ and some $D_a<d_0\le D_{a+1}$,
\[
  d_u^{\le C}+B_a c=B_a c_0+d_0
  \quad\Longrightarrow\quad
  d_u^{\le C}=(c_0-c)B_a+d_0,
\]
where $c=s_{u,a}^{\le C}$ on a high-degree interval and $c=p_u^{\le C}$ on a low-degree interval.
Equality modulo $B_a$ excludes $d_u^{\le C}\le D_a$.
For $d_u^{\le C}>D_{a+1}$, we have $1\le c_0-c\le H_a$ and
\[
d_u^{\le C}\in[(c_0-c)B_a+D_a+1,(c_0-c)B_a+D_{a+1}],
\]
so $\mathsf{Imp}^{a}_{C,j}$ holds by~\eqref{eq:scale-a-impersonation-def}.
The second endpoint has the same formulas with $b,v,\tau$ in place of $a,u,\sigma$.

\paragraph{Estimation.}
Define
\[
  \widehat P^{(a,b)}_{\ell,j,t,\alpha,\beta}(C):=
  \sum_{(\sigma,\tau)\in\Sigma_a\times\Sigma_b}
  \bigl(Y^{\mathrm{LL}}_{\sigma,\tau}(C)-Y^{\mathrm{LU}}_{\sigma,\tau}(C)
    -Y^{\mathrm{UL}}_{\sigma,\tau}(C)+Y^{\mathrm{UU}}_{\sigma,\tau}(C)\bigr),
\]
\[
  \widehat P^{(a,b)}_{\ell,j,t,\alpha,\beta}
  :=\sum_{C\in\Phi_{\le\ell_0}}\widehat P^{(a,b)}_{\ell,j,t,\alpha,\beta}(C).
\]
All expectations below condition on $\omega$ and average over the quantum measurements.
Fix a clause $C$ with $\mathsf{Tgt}^{\ell,j,t,\alpha,\beta}_C=1$ and a pair $(\sigma,\tau)$.
For the Lower-Lower sub-sketch, let $M'-1$ be the size of its live set before the query on $C$, conditioned on remaining active.
The pair-query formula gives
\[
\begin{aligned}
  \E\bigl[Y^{\mathrm{LL}}_{\sigma,\tau}(C)\mid\omega,\text{active before }C\bigr]
  =\frac M{M'}F_{\sigma,\tau}(C)
    \1[N_a(C)\ge L_a(\sigma)]\1[N_b(C)\ge L_b(\tau)].
\end{aligned}
\]
By Lemma~\ref{lem:active_prob}, its probability of remaining active before this query, conditioned on $\omega$, is $M'/M$.
Its contribution is $0$ if it is inactive, so
\[
  \E\bigl[Y^{\mathrm{LL}}_{\sigma,\tau}(C)\mid\omega\bigr]
  =F_{\sigma,\tau}(C)\1[N_a(C)\ge L_a(\sigma)]\1[N_b(C)\ge L_b(\tau)].
\]
Apply the same calculation to the other three sub-sketches, using their respective thresholds and live sets.
Their combined expectation is
\[
\begin{aligned}
  &\E\bigl[Y^{\mathrm{LL}}_{\sigma,\tau}(C)-Y^{\mathrm{LU}}_{\sigma,\tau}(C)
    -Y^{\mathrm{UL}}_{\sigma,\tau}(C)+Y^{\mathrm{UU}}_{\sigma,\tau}(C)\mid\omega\bigr]\\
  &\qquad=F_{\sigma,\tau}(C)\Delta_a(C,\sigma)\Delta_b(C,\tau).
\end{aligned}
\]
If neither endpoint satisfies its impersonation condition and the sampled counts on their actual high-degree intervals are below $\kappa$, there is one matching pair of indices when $a=a(C,j)$ and $b=a(C,t)$.
For that pair, the substituted prefix values satisfy
\[
  (\widehat d_u,\widehat p_u)=(\widetilde d_u^{\le C},\widetilde p_u^{\le C}),\qquad
  (\widehat d_v,\widehat p_v)=(\widetilde d_v^{\le C},\widetilde p_v^{\le C}).
\]
These values give the same pseudo-biases; with signs $s_\alpha,s_\beta$, we obtain $F_{\sigma,\tau}(C)=\1[(\widetilde\alpha_j(C),\widetilde\alpha_t(C))=(\alpha,\beta)]$.
If either interval index is different, all corresponding threshold differences vanish.
Summing over $\sigma,\tau$ gives
\[
\begin{aligned}
  \E\bigl[\widehat P^{(a,b)}_{\ell,j,t,\alpha,\beta}(C)\mid\omega\bigr]
  =\1\bigl[|C|=\ell,\ a(C,j)=a,\ a(C,t)=b,
    (\widetilde\alpha_j(C),\widetilde\alpha_t(C))=(\alpha,\beta)\bigr]
\end{aligned}
\]
under these conditions; both sides are $0$ when $\mathsf{Tgt}^{\ell,j,t,\alpha,\beta}_C=0$.

For every fixed degree interval, at most one index makes its threshold difference nonzero.
On a high-degree interval the ranges
\[
  [B_a\sigma+D_a+1,\ B_a\sigma+D_{a+1}]
\]
are disjoint; on a low-degree interval the levels $B_a c_1+d_1$ are distinct because $B_a>D_{a+1}$.
Since $0\le F_{\sigma,\tau}(C)\le1$, this proves
\[
  0\le\E\bigl[\widehat P^{(a,b)}_{\ell,j,t,\alpha,\beta}(C)\mid\omega\bigr]\le1.
\]
Only impersonation or a sampled count at least $\kappa$ on its actual interval can therefore cause a difference from the coordinate.
Summing over clauses gives
\begin{equation}
\label{eq:two-endpoint-remaining-coordinate-total}
\begin{aligned}
  &\left|\E\bigl[\widehat P^{(a,b)}_{\ell,j,t,\alpha,\beta}\mid\omega\bigr]
    -P^{(a,b)}_{\ell,j,t,\alpha,\beta}\right|\\
  &\le\sum_{C\in\Phi_{\le\ell_0}}
    \mathsf{Tgt}^{\ell,j,t,\alpha,\beta}_C\Bigl(
      \1[\mathsf{Imp}^{a}_{C,j}]+\1[\mathsf{Imp}^{b}_{C,t}]\\
  &\qquad\qquad+\1[a=a(C,j),\ 2D_a\ge\kappa,\ s_{v_j(C),a}^{\le C}\ge\kappa]
      +\1[b=a(C,t),\ 2D_b\ge\kappa,\ s_{v_t(C),b}^{\le C}\ge\kappa]\Bigr).
\end{aligned}
\end{equation}

\subsection{Estimating the weighted pseudo-snapshot score}
\label{sec:proof_of_quantum_algorithm}

In this subsection, we prove Lemma~\ref{lem:maxksat-pairwise-estimator}.

We combine the coordinate estimates from all degree levels. Keep $\kappa=\lceil K\eps^{-7}\rceil$, with $K$ large enough for the overflow bound below.
Choose $K_{\mathrm{base}}$ for the impersonation bound, then $K_\eps$ for the no-wrap tail bounds, and finally $\gamma_\eps$ in $M=\lceil\gamma_\eps m\rceil$ so that queried levels are below $M$ and total shifts are below $M$ on $\mathcal G_1\cap\mathcal G_2$, as required by Lemmas~\ref{lem:good_label_budget} and~\ref{lem:two-endpoint-label-budget}.
After these parameters are fixed, generate the shared public randomness $\omega=(\{f_a\},g,\{B_a\})$.
For fixed $\omega$, all pseudo-snapshot coordinates are deterministic; conditional expectations and variances are over quantum measurements.

Run in parallel, during the same pass over the stream, all sketches for the one-endpoint and two-endpoint coordinates, together with the exact classical counters for $M_\ell$, $3\le\ell\le\ell_0$.
Let $Z$ be their combined output using the score coefficients in Definition~\ref{def:snapshot-score}:
\[
\begin{aligned}
Z
={}&\sum_a\sum_\alpha q_\alpha\widehat U^{(a)}_{1,1,\alpha}
+\sum_{a,b}\sum_{\alpha,\beta}
  \bigl(1-(1-q_\alpha)(1-q_\beta)\bigr)
  \widehat P^{(a,b)}_{2,1,2,\alpha,\beta}\\
&+\sum_{\ell=3}^{\ell_0}\bigg[
  c_\ell M_\ell
  +\sum_{j=1}^\ell\sum_a\sum_\alpha
    u_{\ell,j}(\alpha)\widehat U^{(a)}_{\ell,j,\alpha}\\
&\hspace{5em}+\sum_{1\le j<t\le\ell}\sum_{a,b}\sum_{\alpha,\beta}
    p_{\ell,jt}(\alpha,\beta)
    \widehat P^{(a,b)}_{\ell,j,t,\alpha,\beta}\bigg].
\end{aligned}
\]

Recall that $X_a$ counts all positive copies selected by $f_a$. The events in Lemmas~\ref{lem:good_label_budget} and~\ref{lem:two-endpoint-label-budget} give
\[
\mathcal G:=\mathcal G_1=\mathcal G_2
=\bigcap_{a:\,2D_a\ge\kappa}\{X_a\le K_\eps m/D_a\},\qquad
\Pr_\omega[\mathcal G^c]\le\frac1{160}+\frac1{160}=\frac1{80}.
\]

Recall that $H_s=\kappa$ when $2D_s\ge\kappa$ and $H_s=D_{s+1}$ otherwise, and
\[
\mathsf{Imp}^{s}_{C,r}\iff
d_{v_r(C)}^{\le C}\in\bigcup_{h=1}^{H_s}[hB_s+D_s+1,hB_s+D_{s+1}].
\]
For $C\in\Phi_{\le\ell_0}$ and $r\in[|C|]$, define a set of degree interval indices in $\{0,\ldots,A-1\}$ by
\[
\mathcal J(C,r):=
\{s:D_s<D_{s+1},\ s\neq a(C,r),\ \mathsf{Imp}^{s}_{C,r}\}.
\]
Set $\mathcal J'(C,r):=\{a(C,r)\}\cup\mathcal J(C,r)$, and define
\[
W_{\mathrm{imp}}:=
\sum_{C\in\Phi_{\le\ell_0}}\sum_{r=1}^{|C|}|\mathcal J(C,r)|
=\sum_{C\in\Phi_{\le\ell_0}}\sum_{r=1}^{|C|}
\sum_{\substack{0\le s<A,\ D_s<D_{s+1}\\s\neq a(C,r)}}
\1[\mathsf{Imp}^{s}_{C,r}].
\]
On $\mathcal G$, every nonzero conditional mean contribution at an incorrect degree interval satisfies the impersonation condition and is counted here.
Let $W_{\mathrm{of}}$ count the literal copies at endpoints whose sampled positive-prefix count is at least $\kappa$ in their actual high-degree interval:
\[
W_{\mathrm{of}}:=
\sum_{C\in\Phi_{\le\ell_0}}\sum_{r=1}^{|C|}w_{|C|}\,
\1\!\left[2D_{a(C,r)}\ge\kappa,\
  s_{v_r(C),a(C,r)}^{\le C}\ge\kappa\right].
\]
We first bound the expectations of these counts over $\omega$.

\begin{lemma}[Overflow bound]\label{lem:maxksat-overflow-weight-bound}
\[
\mathbb E_\omega[W_{\mathrm{of}}]\le e^{-\Omega(\kappa)}m.
\]
\end{lemma}

\begin{proof}
Fix $C\in\Phi_{\le\ell_0}$ and $r\in[|C|]$, and let $a=a(C,r)$.
Only intervals with $2D_a\ge\kappa$ contribute. In this case,
\[
s_{v_r(C),a}^{\le C}\sim
\operatorname{Bin}\!\left(p_{v_r(C)}^{\le C},\frac{\kappa}{2D_a}\right),\qquad
\mathbb E_\omega[s_{v_r(C),a}^{\le C}]
\le\frac{\kappa D_{a+1}}{2D_a}\le0.6\kappa
\]
for $0<\eps\le1/100$ and sufficiently large $K$.
A Chernoff bound gives $\Pr_\omega[s_{v_r(C),a}^{\le C}\ge\kappa]\le e^{-\Omega(\kappa)}$.
Equation~\eqref{eq:weighted-total-copies} bounds the total number of literal copies by $m\max_{\ell\in[\ell_0]}\ell w_\ell=O(m)$, so summing these probabilities proves the claim.
\end{proof}

\begin{lemma}[Impersonation-count bound]\label{lem:maxksat-impersonation-weight}
There is a constant $\gamma^\ast_\eps>0$ such that
\[
\mathbb E_\omega[W_{\mathrm{imp}}]
\le\gamma^\ast_\eps\frac{m\log(K_{\mathrm{base}}+2)}{K_{\mathrm{base}}}.
\]
\end{lemma}

\begin{proof}
Fix an endpoint $(C,r)$ with $d:=d_{v_r(C)}^{\le C}$, and an index $s\neq a(C,r)$ with $D_s<D_{s+1}$.
If $\mathsf{Imp}^{s}_{C,r}$ holds, then for some $1\le h\le H_s$,
\[
d\in[hB_s+D_s+1,\ hB_s+D_{s+1}]
\quad\Longleftrightarrow\quad
B_s\in\left[\frac{d-D_{s+1}}h,\frac{d-D_s-1}h\right].
\]
For each $h$, the interval on the right contains $O_\eps(D_{s+1}/h)$ integers, whereas $B_s$ is uniform on $\{\lceil K_{\mathrm{base}}D_{s+1}\rceil,\ldots,\lceil2K_{\mathrm{base}}D_{s+1}\rceil\}$, which contains $\Theta(K_{\mathrm{base}}D_{s+1})$ integers.
Since $H_s=O_\eps(1)$, a union bound gives
\[
\Pr_\omega[\mathsf{Imp}^{s}_{C,r}]
\le\sum_{h=1}^{H_s}O_\eps\!\left(\frac1{K_{\mathrm{base}}h}\right)
=O_\eps\!\left(\frac1{K_{\mathrm{base}}}\right).
\]
The same condition also requires
\[
D_{s+1}<d\le H_s(2K_{\mathrm{base}}D_{s+1}+1)+D_{s+1}.
\]
On high-degree intervals $H_s=\kappa$, so
\[
D_{s+1}\in\left[\frac{d}{2K_{\mathrm{base}}\kappa+O_\eps(1)},\ d\right).
\]
The geometric degree levels give $O_\eps(\log(K_{\mathrm{base}}+2))$ such intervals.
On low-degree intervals $D_{s+1}<\kappa=O_\eps(1)$, so there are only $O_\eps(1)$ intervals.
Thus
\[
\mathbb E_\omega[|\mathcal J(C,r)|]
=\sum_{\substack{0\le s<A,\ D_s<D_{s+1}\\s\neq a(C,r)}}
\Pr_\omega[\mathsf{Imp}^{s}_{C,r}]
=O_\eps\!\left(\frac{\log(K_{\mathrm{base}}+2)}{K_{\mathrm{base}}}\right).
\]
Summing over at most $\ell_0m$ endpoints proves the claim.
\end{proof}

The same count of possible degree intervals gives a deterministic bound: there is a constant $\gamma^{\mathrm{sc}}_\eps>0$ such that
\begin{equation}\label{eq:maxksat-visible-scale-count}
|\mathcal J'(C,r)|\le\gamma^{\mathrm{sc}}_\eps\log(K_{\mathrm{base}}+2)
\qquad\text{for every }(C,r).
\end{equation}

We now fix $\omega\in\mathcal G$ and bound the conditional mean of $Z$.
Lemma~\ref{lem:nonwrapping-threshold-invariant} ensures that each queried label is live exactly when its threshold inequality is true, so the coordinate bounds from Sections~\ref{subsec:One-endpoint_estimation} and~\ref{subsec:Two-endpoint_estimation} apply.

\paragraph{One-endpoint contribution.}
By~\eqref{eq:one-endpoint-high-coordinate}, high-degree one-endpoint outputs have their exact coordinate as conditional mean, and the counters $M_\ell$ are exact.
For a low-degree coordinate, \eqref{eq:one-endpoint-low-coordinate-total} gives
\[
\left|\mathbb E[\widehat U^{(a)}_{\ell,j,\alpha}\mid\omega]-U^{(a)}_{\ell,j,\alpha}\right|
\le\sum_{C\in\Phi_{\le\ell_0}}
\mathsf{Tgt}^{\ell,j,\alpha}_{C}\,\1[\mathsf{Imp}^{a}_{C,j}].
\]
Here $\mathsf{Tgt}^{\ell,j,\alpha}_{C}=\1[|C|=\ell,\ s_j(C)=s_\alpha]$.
Each $(C,r,a)$ with $a\in\mathcal J(C,r)$ occurs for at most $L$ bucket labels.
The fixed score coefficients therefore bound the total one-endpoint error by $\gamma^{\mathrm{one}}_\eps W_{\mathrm{imp}}$ for a constant $\gamma^{\mathrm{one}}_\eps>0$.

\paragraph{Two-endpoint error from incorrect interval pairs.}
In the endpoint formulas of Section~\ref{subsec:Two-endpoint_estimation}, a nonzero threshold difference requires its interval index to belong to $\mathcal J'(C,r)$.
Consequently, for fixed $C$ with $|C|=\ell$, positions $1\le j<t\le\ell$, and labels $\alpha,\beta$,
\[
\mathbb E[\widehat P^{(a,b)}_{\ell,j,t,\alpha,\beta}(C)\mid\omega]\neq0
\quad\Longrightarrow\quad
(a,b)\in\mathcal J'(C,j)\times\mathcal J'(C,t).
\]
When $(a,b)\neq(a(C,j),a(C,t))$, the target contribution of $C$ is zero and at least one queried index belongs to $\mathcal J$.
The endpoint formulas in Section~\ref{subsec:Two-endpoint_estimation} give
\[
0\le\mathbb E[\widehat P^{(a,b)}_{\ell,j,t,\alpha,\beta}(C)\mid\omega]\le1.
\]
By~\eqref{eq:maxksat-visible-scale-count}, the number of these pairs is at most
\[
\begin{aligned}
&|\mathcal J(C,j)|\,|\mathcal J'(C,t)|
 +|\mathcal J(C,t)|\,|\mathcal J'(C,j)|\\
&\qquad\le\gamma^{\mathrm{sc}}_\eps\log(K_{\mathrm{base}}+2)
\bigl(|\mathcal J(C,j)|+|\mathcal J(C,t)|\bigr).
\end{aligned}
\]
Summing over clauses, positions, and labels with the fixed score coefficients bounds the total error from incorrect interval pairs by
\[
\gamma^{\mathrm{two}}_\eps\log(K_{\mathrm{base}}+2)\,W_{\mathrm{imp}}
\]
for a constant $\gamma^{\mathrm{two}}_\eps>0$.

\paragraph{Overflow contribution.}
The incorrect interval pairs have already been counted above.
For $|C|=\ell$, $a=a(C,j)$, and $b=a(C,t)$, the overflow error in~\eqref{eq:two-endpoint-remaining-coordinate-total} satisfies
\[
\begin{aligned}
&\left|\mathbb E[\widehat P^{(a,b)}_{\ell,j,t,\alpha,\beta}(C)\mid\omega]
-\1[(\widetilde\alpha_j(C),\widetilde\alpha_t(C))=(\alpha,\beta)]\right|\\
&\quad\le\1[2D_a\ge\kappa,\ s_{v_j(C),a}^{\le C}\ge\kappa]
+\1[2D_b\ge\kappa,\ s_{v_t(C),b}^{\le C}\ge\kappa].
\end{aligned}
\]
Each endpoint occurs in only constantly many position and label choices, and $w_\ell\ge1$.
Thus the fixed score coefficients bound this remaining error by $\gamma_{\mathrm{of}}W_{\mathrm{of}}$ for a constant $\gamma_{\mathrm{of}}>0$.
The constants $\gamma^\ast_\eps$, $\gamma^{\mathrm{sc}}_\eps$, $\gamma^{\mathrm{one}}_\eps$, and $\gamma^{\mathrm{two}}_\eps$ may depend on $\eps$, $\kappa$, the fixed rounding profile and coefficients, but not on $K_{\mathrm{base}}$; $\gamma_{\mathrm{of}}$ depends only on the fixed score coefficients.
Together, these bounds give a constant $\gamma_{\mathrm{imp}}=O_\eps(\log(K_{\mathrm{base}}+2))$, with the implicit constant independent of $K_{\mathrm{base}}$, such that for every $\omega\in\mathcal G$,
\begin{equation}\label{eq:maxksat-one-run-mean}
\left|\mathbb E[Z\mid\omega]
-L_{\le\ell_0}^\Theta(\PsSnap(\Phi_{\le\ell_0}))\right|
\le\gamma_{\mathrm{imp}}W_{\mathrm{imp}}+\gamma_{\mathrm{of}}W_{\mathrm{of}}.
\end{equation}

The choices of $K_{\mathrm{base}}$ and $K$ above can ensure, by Lemmas~\ref{lem:maxksat-impersonation-weight} and~\ref{lem:maxksat-overflow-weight-bound}, that
\[
\gamma_{\mathrm{imp}}\mathbb E_\omega[W_{\mathrm{imp}}]\le\frac{\eps m}{320},\qquad
\gamma_{\mathrm{of}}\mathbb E_\omega[W_{\mathrm{of}}]\le\frac{\eps m}{320}.
\]
Indeed, the first bound is $O_\eps(m\log^2(K_{\mathrm{base}}+2)/K_{\mathrm{base}})$, and the second is $e^{-\Omega(\kappa)}m$ times a fixed constant.
Define
\[
\mathcal S:=\left\{W_{\mathrm{imp}}\le\frac{\eps m}{4\gamma_{\mathrm{imp}}}\right\},\qquad
\mathcal O:=\left\{W_{\mathrm{of}}\le\frac{\eps m}{4\gamma_{\mathrm{of}}}\right\}.
\]
Markov's inequality gives $\Pr_\omega[\mathcal S^c],\Pr_\omega[\mathcal O^c]\le1/80$.
For every $\omega\in\mathcal G\cap\mathcal S\cap\mathcal O$, Equation~\eqref{eq:maxksat-one-run-mean} gives
\begin{equation}\label{eq:maxksat-one-run-mean-small}
\left|\mathbb E[Z\mid\omega]
-L_{\le\ell_0}^\Theta(\PsSnap(\Phi_{\le\ell_0}))\right|
\le\frac{\eps m}{2}.
\end{equation}

There are $O_\eps(A)$ one-endpoint coordinates and $O_\eps(A^2)$ two-endpoint coordinates, with $A=O_\eps(\log(n+m))$.
Each coordinate combines $O_\eps(1)$ sub-sketches, whose outputs have absolute value at most $M=O_\eps(m)$.
For fixed $\omega$, different coordinates use separate work registers and have independent measurement outputs. Therefore
\begin{equation}\label{eq:maxksat-one-run-var}
\operatorname{Var}(Z\mid\omega)=O_\eps(m^2\log^2(n+m))
\qquad(\omega\in\mathcal G).
\end{equation}

Run $R=\gamma^{\mathrm{rep}}_\eps\log^2(n+m)$ copies in parallel with separate work registers and the same $\omega$, and set
\[
Z':=\frac1R\sum_{r=1}^R Z_r.
\]
Their outputs are independent conditional on $\omega$, so for every $\omega\in\mathcal G$,
\[
\mathbb E[Z'\mid\omega]=\mathbb E[Z\mid\omega],\qquad
\operatorname{Var}(Z'\mid\omega)
=\frac1R\operatorname{Var}(Z\mid\omega)
=\frac{O_\eps(m^2)}{\gamma^{\mathrm{rep}}_\eps}.
\]
For sufficiently large $\gamma^{\mathrm{rep}}_\eps$, Chebyshev's inequality gives
\[
\Pr\left[\left|Z'-\mathbb E[Z\mid\omega]\right|>\frac{\eps m}{4}
\ \middle|\ \omega\right]\le\frac1{40}
\qquad(\omega\in\mathcal G).
\]
For $\omega\in\mathcal G\cap\mathcal S\cap\mathcal O$, Equation~\eqref{eq:maxksat-one-run-mean-small} gives
\[
\left|Z'-L_{\le\ell_0}^\Theta(\PsSnap(\Phi_{\le\ell_0}))\right|
\le\left|Z'-\mathbb E[Z\mid\omega]\right|+\frac{\eps m}{2}.
\]
Averaging the conditional probability over $\omega$ and taking a union bound gives
\[
\begin{aligned}
&\Pr\left[\left|Z'-L_{\le\ell_0}^\Theta(\PsSnap(\Phi_{\le\ell_0}))\right|>\eps m\right]
\le\Pr_\omega[\mathcal G^c]+\Pr_\omega[\mathcal S^c]+\Pr_\omega[\mathcal O^c]+\frac1{40}\le\frac1{16}.
\end{aligned}
\]
Thus $\widehat L_{\le\ell_0}:=Z'$ succeeds with probability at least $15/16$.

\paragraph{Space complexity.}
A high-degree one-endpoint coordinate uses $2\kappa+2=O_\eps(1)$ sub-sketches.
On a low-degree interval, $D_{a+1}=O_\eps(1)$, so there are $O_\eps(1)$ candidate pairs $(d,c)$.
A two-endpoint coordinate uses $4|\Sigma_a||\Sigma_b|=O_\eps(1)$ sub-sketches; $\Sigma_a,\Sigma_b$ contain sampled counts on high-degree intervals and candidate pairs $(d,c)$ on low-degree intervals.
Each sub-sketch uses $O(\log(n+m))$ qubits and, after a successful query, $O(\log(n+m))$ classical bits to store its variables, indices, and suffix degrees.
The $O_\eps(A^2)$ coordinates and $R=O_\eps(\log^2(n+m))$ repetitions therefore use
\[
O_\eps\!\bigl(A^2R\log(n+m)\bigr)=O_\eps\!\bigl(\log^5(n+m)\bigr)
\]
qubits and the same order of classical bits.

\section{Quantum algorithm for Max-\texorpdfstring{$k$}{k}SAT}
\label{subsec:maxksat-final-algorithm}
In this section, we put everything together and  prove Theorem~\ref{thm:main}. Moreover, we   combine the  pseudo-snapshot estimation with its score guarantee to obtain the $0.7426$-approximation.

\begin{proof}[Proof of Theorem~\ref{thm:main}]
Take the rounding profile and coefficients from the snapshot score bounds in Lemma~\ref{lem:maxksat-weighted-snapshot-bound}, with $\rho=0.742694813$, and the additive-error constants $\gamma_0,\gamma_1$ from the pseudo-snapshot score guarantee in Theorem~\ref{thm:weighted-pseudosnapshot-to-value}.
Choose $0<\eps\le1/100$ so that this score guarantee holds with probability at least $15/16$ and
\[
  2(\gamma_1+2)\eps\le\rho-0.7426.
\]

Recall $D_0=0$, $D_a=\lfloor(1+\eps^3)^a\rfloor$ for $a\ge1$, and $\kappa=\lceil K\eps^{-7}\rceil$.
The bits $f_a$ sample positive copies independently with probability $\kappa/(2D_a)$ when $2D_a\ge\kappa$, and $g(v)\sim\Unif[-\eps,\eps]$ is the noise shared by all clauses containing $x_v$.
Independently of all sampling bits and $g$, each $B_a$ is uniform on $\{\lceil K_{\mathrm{base}}D_{a+1}\rceil,\ldots,\lceil2K_{\mathrm{base}}D_{a+1}\rceil\}$ and enters the prefix counts as $\Gamma_a(d,c)=d+B_ac$.
All estimators in a trial share $\omega=(\{f_a\},g,\{B_a\})$.

\paragraph{Algorithm.}
The input is $\Phi_{\mathrm{orig}}=(C_1,\ldots,C_m)$ and $\delta\in(0,1)$.
As each clause arrives, remove repeated literals and discard clauses containing both a variable and its negation. Let $\Phi$ be the remaining instance and
\[
T:=\#\{\nu\in[m]:\exists v\in[n],\ \{x_v,\neg x_v\}\subseteq C_\nu\}.
\]
Maintain $M_{>\ell_0}=\#\{C\in\Phi:|C|>\ell_0\}$ exactly and run $R_{\mathrm{out}}=\Theta(\log(2/\delta))$ independent trials in parallel:
\begin{enumerate}[leftmargin=1.5em,itemsep=0.2em,topsep=0.2em]
  \item In each trial, sample independent public randomness $\omega$ and run the score estimator with $\eps m$ error from Lemma~\ref{lem:maxksat-pairwise-estimator} on $\Phi_{\le\ell_0}$, obtaining $\widehat L_{\le\ell_0}$.
  \item After the stream, each trial outputs
  \[
    Y:=\max\left\{0,\widehat L_{\le\ell_0}
      +\rho M_{>\ell_0}-(\gamma_0+1)\eps m\right\}.
  \]
  \item Return $Z:=T+Z_0$, where $Z_0:=\operatorname{median}(Y_1,\ldots,Y_{R_{\mathrm{out}}})$ and $Y_r$ is the output of trial $r$.
\end{enumerate}

\paragraph{Approximation guarantee.}
For one trial, Theorem~\ref{thm:weighted-pseudosnapshot-to-value} and Lemma~\ref{lem:maxksat-pairwise-estimator} give, simultaneously with probability at least $7/8$,
\[
  \rho\,\OPT(\Phi)-\gamma_1\eps m
  \le L_{\le\ell_0}^\Theta(\PsSnap(\Phi_{\le\ell_0}))
    +\rho M_{>\ell_0}-\gamma_0\eps m
  \le\OPT(\Phi),
\]
\[
  \left|\widehat L_{\le\ell_0}
    -L_{\le\ell_0}^\Theta(\PsSnap(\Phi_{\le\ell_0}))\right|\le\eps m.
\]
Substituting the second inequality into the first yields
\[
  \rho\,\OPT(\Phi)-(\gamma_1+2)\eps m\le Y\le\OPT(\Phi).
\]
A Chernoff bound shows that this interval contains a strict majority of the trial outputs with probability at least $1-\delta$, and hence contains $Z_0$.
Since $\OPT(\Phi_{\mathrm{orig}})=T+\OPT(\Phi)$,
\[
  \rho\,\OPT(\Phi_{\mathrm{orig}})-(\gamma_1+2)\eps m
  \le T+\rho\,\OPT(\Phi)-(\gamma_1+2)\eps m
  \le Z\le\OPT(\Phi_{\mathrm{orig}}).
\]
A uniformly random assignment satisfies each input clause with probability at least $1/2$, so $\OPT(\Phi_{\mathrm{orig}})\ge m/2$.
Our choice of $\eps$ gives
\[
  0.7426\,\OPT(\Phi_{\mathrm{orig}})\le Z\le\OPT(\Phi_{\mathrm{orig}}).
\]

\paragraph{Space.}
Each trial uses $O_\eps(\log^5(n+m))$ qubits and the same number of classical bits by Lemma~\ref{lem:maxksat-pairwise-estimator}; the exact counters use $O(\log m)$ classical bits.
Thus the total space is
\[
  O\!\left(\log^5(n+m)\log\frac2\delta\right)
\]
qubits and classical bits.
The public random bits are not charged to the working space.
For $m=n^{O(1)}$ and $\delta\le1/2$, this is $O(\log^5 n\log(1/\delta))$; for larger $\delta$, a constant number of trials suffices.
\end{proof}

\bibliographystyle{alpha}
\bibliography{ref}

\appendix
\section{Quantum space lower bound for Max-\texorpdfstring{$k$}{k}SAT}
\label{app:maxksat-lower-bound}

In this section, we prove the quantum space lower bound in Theorem~\ref{thm:maxksat-quantum-lb} by reducing the Max-Cut promise problem to Max-$k$SAT.

For a graph $G=(V,E)$ with $m_E:=|E|$ and an assignment $\sigma:V\to\{0,1\}$, define
\[
  \Cut_G(\sigma)
  :=
  \bigl|\{\{u,v\}\in E:\sigma(u)\ne\sigma(v)\}\bigr|,
  \qquad
  \MaxCut(G):=\max_\sigma \Cut_G(\sigma).
\]

We use the following promise-gap form of the quantum streaming lower bound for Max-Cut of Kallaugher and Parekh~\cite{KP22}.

\begin{theorem}\cite{KP22}\label{max-cut}
For every constant $0<\varepsilon<1/2$, any one-pass quantum streaming algorithm that, on an $n$-vertex graph $G$ with $m_E$ edges, distinguishes
\[
  \MaxCut(G)=m_E
  \qquad\text{from}\qquad
  \MaxCut(G)\le (1/2+\varepsilon)m_E
\]
with constant advantage requires $\Omega(n)$ qubits of space.
\end{theorem}

\begin{proof}[Proof of Theorem \ref{thm:maxksat-quantum-lb}]
We reduce from this Max-Cut promise problem. Given the edge stream of $G$, we construct a Max-$k$-SAT stream $\Phi_G$ online by replacing each edge $\{u,v\}$ with the two clauses
\[
  (x_u\vee x_v),
  \qquad
  (\neg x_u\vee \neg x_v).
\]
Since $k\ge2$, these are valid Max-$k$-SAT clauses.  The reduction is one-pass and online, and uses only $O(\log (n+m_E))$ additional classical space. The resulting instance has $M=2m_E$ clauses.  In what follows, $m_E$ always denotes the number of Max-Cut edges, not the number of Max-$k$SAT clauses.

For every assignment $\sigma$, each edge contributes
$1+\mathbf 1[\sigma(u)\ne\sigma(v)]$ satisfied clauses: if $\sigma(u)\ne\sigma(v)$, both clauses are satisfied, whereas if $\sigma(u)=\sigma(v)$, exactly one is satisfied.  Summing over all edges gives
\[
  \Val_{\Phi_G}(\sigma)=m_E+\Cut_G(\sigma).
\]
Taking the maximum over $\sigma$, we obtain
$\OPT(\Phi_G)=m_E+\MaxCut(G)$.

Now suppose, toward contradiction, that there is a one-pass quantum streaming algorithm $\mathcal A$ using $o(n)$ qubits and achieving an approximation ratio
$3/4+\gamma$ for Max-$k$-SAT, for a fixed $\gamma>0$, with error probability $\delta=1/3$. Choose $0<\varepsilon<\min\{2\gamma,1/2\}$, and run $\mathcal A$ on the stream $\Phi_G$ produced above. Let $Z:=\mathcal A(\Phi_G)$; with probability at least $2/3$,
\[
  \left(\frac34+\gamma\right)\OPT(\Phi_G)\le Z\le\OPT(\Phi_G).
\]

If $\MaxCut(G)=m_E$, then
\[
  \OPT(\Phi_G)=2m_E=M,
\]
so a successful run of $\mathcal A$ outputs
\[
  Z\ge
  \left(\frac34+\gamma\right)2m_E
  =
  \left(\frac32+2\gamma\right)m_E.
\]
If instead $\MaxCut(G)\le (1/2+\varepsilon)m_E$, then
\[
  \OPT(\Phi_G)
  \le
  m_E+\left(\frac12+\varepsilon\right)m_E
  =
  \left(\frac32+\varepsilon\right)m_E,
\]
and a successful run satisfies $Z\le\OPT(\Phi_G)$.

Output the first case if $Z>(3/2+\varepsilon)m_E$ and the second case otherwise. Since $\varepsilon<2\gamma$, this distinguishes the Max-Cut promises with probability at least $2/3$ using $o(n)$ space, contradicting Theorem~\ref{max-cut}.  Therefore, any one-pass quantum streaming algorithm beating $3/4$ for Max-$k$SAT requires $\Omega(n)$ qubits of space.
\end{proof}

\section{Snapshot score approximation}
\label{sec:Snapshot_score_approximation}
In this section, we prove Lemma~\ref{lem:maxksat-weighted-snapshot-bound}. Given a rounding profile $\Theta= (\ell_0,\bw,\mathcal{I},\br)$, 
we first define several linear constraints on the coefficients $\{c_{\ell }, u_{\ell,j}(\alpha), p_{\ell,jt}(\alpha,\beta)\}$ and combine them to formulate a linear programming, denoted by $\mathsf{LP}(\Theta)$,  that  maximizes the approximation $\rho$.
Finally, an AI agent searches for a good  rounding profile $\Theta$, and exact verification gives $\rho=0.742694813$.

\theoremstyle{plain}
\newtheorem*{restatedsnapshotbound}{Lemma~\ref{lem:maxksat-weighted-snapshot-bound}}
\theoremstyle{definition}

\begin{restatedsnapshotbound}
There exist  a rounding profile $\Theta= (\ell_0,\bw,\mathcal{I},\br)$ and coefficients $ \{c_{\ell
 }, u_{\ell,j}(\alpha), p_{\ell,jt}(\alpha,\beta) \}$ where $3\le \ell\le \ell_0$, $j\in[\ell]$ for $u_{\ell,j}$, $1\le j<t\le\ell$ for $p_{\ell,jt}$, and $\alpha,\beta\in \{+,-\}\times [L]$ so that  for $\rho = 0.742694813$,  we have
  \[
  1-\left(\max_{\alpha}(1-q_\alpha)\right)^{\ell_0+1}\ge\rho,
 \]
 and
 \[
 \rho\cdot \OPT(\Phi_{\le \ell_0})\le  L_{\le\ell_0}^\Theta(\Snap(\Phi_{\le\ell_0})) \le \mathcal{E}^{\Theta}(\Phi_{\le\ell_0}).
 \]
\end{restatedsnapshotbound}

\begin{proof}[Proof of Lemma~\ref{lem:maxksat-weighted-snapshot-bound}]
\leavevmode\par
\paragraph{Step 1. Define the first constraint.}
Fix a rounding profile $\Theta = (\ell_0,\bw,\mathcal{I},\br)$. For $3\le\ell\le\ell_0$, define the score of a clause with labels $\alpha_1,\ldots,\alpha_\ell$ by
\[
  H_\Theta(\alpha_1,\ldots,\alpha_\ell)
  :=c_\ell+\sum_{j=1}^\ell u_{\ell,j}(\alpha_j)
    +\sum_{1\le j<t\le\ell}p_{\ell,jt}(\alpha_j,\alpha_t).
\]
The first constraint requires
\begin{equation}
  H_\Theta(\alpha_1,\ldots,\alpha_\ell)
  \le1-\prod_{j=1}^\ell(1-q_{\alpha_j})
  \label{eq:snapshot-score-upper}
\end{equation}
for every $3\le\ell\le\ell_0$ and every choice of labels.
This constraint ensures that the snapshot score is upper bounded by  the expected number of satisfied clauses under the random rounding. 
Clauses of length one or two already use their exact satisfaction probabilities.

\paragraph{Step 2. Define the second constraint.}
The second constraint is to  guarantee the approximation ratio $\rho$ of  the snapshot score.

At first, we need some new definitions and notations. 
Fix an optimal assignment $\bx^*$ for $\Phi_{\le\ell_0}$. Let $i(v)$ be the bucket index with bias $\widetilde b_v\in I_{i(v)}$. Group the clauses of $\Phi_{\le\ell_0}$ by their length, literal signs, bucket indices, and variable values under $\bx^*$; and we define the \textit{record} by 
\[
  c=(\ell;s_1,\ldots,s_\ell;i_1,\ldots,i_\ell;\tau_1,\ldots,\tau_\ell),
  \qquad \ell\in[\ell_0],\quad s_j\in\sgnset,\quad i_j\in[L],\quad\tau_j\in\bits.
\]
For each clause $C$, these entries are $s_j=s_j(C)$, $i_j=i(v_j(C))$, and $\tau_j=x^*_{v_j(C)}$. Define the number of clauses with record $c$ by 
$$
n_c := \#\{C\in\Phi_{\le\ell_0}:C\text{ has record } c \text{ under } \mathbf{x}^*\}.
$$

Write the label $\alpha_j=(s_j,i_j)$ and define the contribution of a record to the snapshot score and the optimal value by
\[
  h_\Theta(c):=\begin{cases}
    q_{\alpha_1},&\ell=1,\\
    1-(1-q_{\alpha_1})(1-q_{\alpha_2}),&\ell=2,\\
    H_\Theta(\alpha_1,\ldots,\alpha_\ell),&3\le\ell\le\ell_0,
  \end{cases}
\]
\[
  o(c):=\1\bigl[\exists j\in[\ell]:\ (s_j,\tau_j)\in\{(+,1),(-,0)\}\bigr].
\]
Therefore, 
summing over all records gives
\begin{equation}
  L_{\le\ell_0}^\Theta(\Snap(\Phi_{\le\ell_0}))=\sum_c n_ch_\Theta(c),
  \qquad \OPT(\Phi_{\le\ell_0})=\sum_c n_co(c).
  \label{eq:snapshot-record-sums}
\end{equation}

The bucket bounds restrict the counts $n_c$.
Let $a_i,b_i$ be the endpoints of $I_i$.
Put $\chi(+)=1$ and $\chi(-)=-1$. For $i\in[L]$ and $\tau\in\bits$, define the contribution of a record to the lower and upper bucket bounds by
\[
  g^{\rm lo}_{i,\tau}(c):=
  w_\ell\sum_{\substack{j\in[\ell]:\,\\
  i_j=i,  \tau_j=\tau}}(a_i-\chi(s_j)),
  \qquad
  g^{\rm up}_{i,\tau}(c):=
  w_\ell\sum_{\substack{j\in[\ell]:\\ i_j=i,  \tau_j=\tau } }(\chi(s_j)-b_i),
\]
where $w_\ell$, the $\ell$-th entry of $\bw$, is the weight of clauses with length $\ell$.
Their sums over the records are nonpositive, as the following claim shows.
\begin{claim}\label{clm:snapshot-bucket-sums}
For every $i\in[L]$ and $\tau\in\bits$,
\[
  \sum_c n_cg^{\rm lo}_{i,\tau}(c)\le0,
  \qquad \sum_c n_cg^{\rm up}_{i,\tau}(c)\le0.
\]
\end{claim}
\begin{proof}
Within the variables with $i(v)=i$ and $x_v^*=\tau$, let
\[
  d_{i,\tau}:=\sum_{\substack{v:\,i(v)=i, x_v^*=\tau}}(\widetilde p_v+\widetilde q_v),
  \qquad
  e_{i,\tau}:=\sum_{\substack{v:\,i(v)=i, x_v^*=\tau}}(\widetilde p_v-\widetilde q_v).
\]
Since $a_i\le\widetilde b_v\le b_i$ and $\widetilde p_v-\widetilde q_v=\widetilde b_v(\widetilde p_v+\widetilde q_v)$,
$
  a_id_{i,\tau}\le e_{i,\tau}\le b_id_{i,\tau};
$
each literal in a clause of length $\ell$ contributes $w_\ell$ to $d_{i_j,\tau_j}$ and $w_\ell\chi(s_j)$ to $e_{i_j,\tau_j}$. Hence
\[
  \sum_c n_cg^{\rm lo}_{i,\tau}(c)=a_id_{i,\tau}-e_{i,\tau}\le0,
  \qquad
  \sum_c n_cg^{\rm up}_{i,\tau}(c)=e_{i,\tau}-b_id_{i,\tau}\le0. \qedhere
\]
\end{proof}

\medskip
The second constraint is to   find    coefficients $\lambda^{\rm lo}_{i,\tau}\ge 0,\lambda^{\rm up}_{i,\tau}\ge 0$ such that for   every record $c$,
\begin{equation}
  h_\Theta(c)-\rho\cdot  o(c)
  +\sum_{i,\tau}\bigl(
    \lambda^{\rm lo}_{i,\tau}g^{\rm lo}_{i,\tau}(c)
    +\lambda^{\rm up}_{i,\tau}g^{\rm up}_{i,\tau}(c)\bigr)\ge0.
  \label{eq:snapshot-record-lower}
\end{equation}
This constraint  builds on the bucket bounds and allows more flexibility in the coefficients so that we can achieve a better approximation ratio $\rho$. Without this  relaxation,  if  we simply ask for   
\[
 h_\Theta(c)-\rho\cdot  o(c)\ge 0,
\] the satisfied unary records $(1;\,+;\,i;\,1)$ and $(1;\,-;\,i;\,0)$ would require $r_i\ge\rho$ and $1-r_i\ge\rho$, forcing $\rho\le1/2$.

\paragraph{Step 3. Maximize the approximation ratio.}
Given the rounding profile  $\Theta$, we define a linear program $\mathsf{LP}(\Theta)$, combining the constraints from the last two steps, to maximize the approximation ratio $\rho$; the variables of $\mathsf{LP}(\Theta)$ are $\rho$, the score coefficients $c_\ell,u_{\ell,j}(\alpha),p_{\ell,jt}(\alpha,\beta)$, and the extra  coefficients $\lambda^{\rm lo}_{i,\tau},\lambda^{\rm up}_{i,\tau}$ from Step 2.  
$\mathsf{LP}(\Theta)$ is defined as follows: 
\[
\begin{aligned}
  \text{maximize}\quad &\rho\\
  \text{s.t.}\quad
  & (1).\;H_\Theta(\alpha_1,\ldots,\alpha_\ell)
    \le1-\prod_{j=1}^\ell(1-q_{\alpha_j})
    &&(3\le\ell\le\ell_0;\text{ all labels}),\\
  &(2).\; h_\Theta(c)-\rho o(c)
    +\sum_{i,\tau}\bigl(\lambda^{\rm lo}_{i,\tau}g^{\rm lo}_{i,\tau}(c)
    +\lambda^{\rm up}_{i,\tau}g^{\rm up}_{i,\tau}(c)\bigr)\ge0
    &&\text{for every record }c,\\
  &(3).\; \lambda^{\rm lo}_{i,\tau},\lambda^{\rm up}_{i,\tau}\ge0
    &&(i\in[L],\ \tau\in\bits),\\
  &(4).\; 0\le\rho\le 1-\left(\max_\alpha(1-q_\alpha)\right)^{\ell_0+1}.
\end{aligned}
\]
Once $\Theta$ is fixed, all $q_\alpha$ and $g^{\rm lo}_{i,\tau}(c),g^{\rm up}_{i,\tau}(c)$ are constants. The first two constraint families are therefore linear, and their index sets are finite. The constraint $(4)$ ensures that every clause of length greater than $\ell_0$ is satisfied with probability at least $\rho$ under this rounding.

\paragraph{Step 4. Prove the approximation guarantee.}
We first  prove the upper bound. Take any feasible solution of $\mathsf{LP}(\Theta)$. 
For each $3\le\ell\le\ell_0$, grouping clauses by their labels gives
\begin{equation}
  \label{eq:the_2nd_part_score}
  \begin{aligned}
  \sum_{C\in\Phi: |C|=\ell} H_\Theta(\alpha_1(C),\ldots,\alpha_\ell(C))
  =c_\ell M_\ell+\!\sum_{j=1}^\ell\sum_\alpha u_{\ell,j}(\alpha)U_{\ell,j,\alpha}
  +\sum_{1\le j<t\le\ell}\sum_{\alpha,\beta}
    p_{\ell,jt}(\alpha,\beta)P_{\ell,j,t,\alpha,\beta}.
\end{aligned}
\end{equation}
Under the  random  rounding by $\Theta$, the expected number of satisfied clauses is
\[
\begin{aligned}
  \mathcal E^\Theta(\Phi_{\le\ell_0})
  &=\sum_{\ell=1}^{\ell_0}\sum_{C\in\Phi:\,|C|=\ell}\E[C(\bx)]
  =\sum_{\ell=1}^{\ell_0}\sum_{C\in\Phi:\,|C|=\ell}
    \left(1-\prod_{j=1}^\ell(1-q_{\alpha_j(C)})\right)\\
  &\ge \mathcal{E}^{\Theta}(\Phi_{\le 2})
    +\sum_{\ell=3}^{\ell_0}\sum_{C\in\Phi:\,|C|=\ell}
      H_\Theta(\alpha_1(C),\ldots,\alpha_\ell(C))\\
  &=L_{\le\ell_0}^\Theta(\Snap(\Phi_{\le\ell_0})).
\end{aligned}
\]
The second last inequality is by the constraint  (1), and the last equality is by Equation~\eqref{eq:expect_2_term} and Equality~\eqref{eq:the_2nd_part_score}.

For the lower bound, multiply Inequality~\eqref{eq:snapshot-record-lower} by $n_c$ and sum over records $c$; by  Equality~\eqref{eq:snapshot-record-sums}, we have
\[
\begin{aligned}
  L_{\le\ell_0}^\Theta(\Snap(\Phi_{\le\ell_0}))-\rho\,\OPT(\Phi_{\le\ell_0})+ \sum_{i,\tau} \underbrace{  \left(
    \lambda^{\rm lo}_{i,\tau}\sum_c n_cg^{\rm lo}_{i,\tau}(c)
    +\lambda^{\rm up}_{i,\tau}\sum_c n_cg^{\rm up}_{i,\tau}(c)\right)}_{\le 0}  \ge 0.
\end{aligned}
\]
where the underbrace bound $\le 0$ follows from the nonpositivity of the sums in Claim~\ref{clm:snapshot-bucket-sums} and $\lambda^{\rm lo}_{i,\tau},\lambda^{\rm up}_{i,\tau}\ge 0$. This proves \[L_{\le\ell_0}^\Theta(\Snap(\Phi_{\le\ell_0}))-\rho\,\OPT(\Phi_{\le\ell_0})\ge 0.\]

\paragraph{Step 5. Select and verify the rounding profile.}
An AI agent alternates between updating $\Theta$ and solving $\mathsf{LP}(\Theta)$ to improve $\rho$.

We convert the final coefficients to rational numbers and verify every constraint using exact arithmetic. The resulting rounding profile has $\ell_0=5$, $L=1000$, and for any $i\in [1000]$,  $61/250\le r_i\le189/250$. It gives a feasible solution with
\[
  \rho=0.742694813,
  \qquad 1-\left(\max_\alpha(1-q_\alpha)\right)^{\ell_0+1} =  1-\left(\frac{189}{250}\right)^6>
  \rho.
\]
The code for selecting and verifying this rounding profile is available at \url{https://github.com/Guangxu-Yang/maxksat-lp-certificate}.
\end{proof}

\section{Proofs for the weighted snapshot and pseudo-snapshot}
\label{app:weighted-snapshot-proofs}

In this section, we prove Lemmas~\ref{lem:weighted-snapshot-score-error} and~\ref{lem:weighted-snapshot-closeness}, i.e., the additive score bound and the snapshot distance bound.

Recall that $\Snap_{\eps}$ uses labels $(s_j(C),i'(v_j(C)))$, where
\[
  b'_v:=\clip(\widetilde b_v+g(v)),\qquad b'_v\in I_{i'(v)}.
\]
Here $\clip(z)=\max\{-1,\min\{1,z\}\}$, $g(v)$ is the same noise used in the pseudo-snapshot, and $a_i,b_i$ are the endpoints of $I_i$.
Since $\widetilde b_v\in[-1,1]$ and $|g(v)|\le\eps$,
\[
  |b'_v-\widetilde b_v|\le\eps
  \quad\Longrightarrow\quad
  a_{i'(v)}-\eps\le\widetilde b_v\le b_{i'(v)}+\eps.
\]
\theoremstyle{plain}
\newtheorem*{restatedscorelemma}{Lemma~\ref{lem:weighted-snapshot-score-error}}
\newtheorem*{restatedclosenesslemma}{Lemma~\ref{lem:weighted-snapshot-closeness}}
\theoremstyle{definition}

\begin{restatedscorelemma}[Snapshot score with an additive error]
Use the rounding profile $\Theta$ and coefficients giving the snapshot score approximation in Lemma~\ref{lem:maxksat-weighted-snapshot-bound}.
There is a constant $\gamma_2>0$, depending only on the fixed rounding profile and coefficients, such that for every choice of $g(v)\in[-\eps,\eps]$, $v\in[n]$, we have
\[
  \rho\,\OPT(\Phi)-\gamma_2\eps m
  \le L_{\le\ell_0}^\Theta(\Snap_{\eps}(\Phi_{\le\ell_0}))
     +\rho M_{>\ell_0}
  \le\OPT(\Phi).
\]
\end{restatedscorelemma}

We restate the LP inequalities used for each bound and apply them with the indices $i'(v)$.
\begin{proof}[Proof of Lemma~\ref{lem:weighted-snapshot-score-error}]
\leavevmode\par
\paragraph{Step 1. Prove the upper bound.}
For every $3\le\ell\le\ell_0$ and every choice of labels, constraint~\eqref{eq:snapshot-score-upper} gives
\[
  H_\Theta(\alpha_1,\ldots,\alpha_\ell)
  =c_\ell+\sum_{j=1}^\ell u_{\ell,j}(\alpha_j)
   +\sum_{1\le j<t\le\ell}p_{\ell,jt}(\alpha_j,\alpha_t)
  \le1-\prod_{j=1}^\ell(1-q_{\alpha_j}).
\]
Choose $\bx\in\bits^n$ by setting its coordinates independently, with $\Pr[x_v=1]=r_{i'(v)}$ for every $v\in[n]$.
For each clause $C$,
\[
  \E[C(\bx)]=1-\prod_{j=1}^{|C|}(1-q_{(s_j(C),i'(v_j(C)))}).
\]
Clauses of length one or two use this probability as their score, so summing the clause scores gives
\[
  L_{\le\ell_0}^\Theta(\Snap_{\eps}(\Phi_{\le\ell_0}))
  \le\E[\Val_{\Phi_{\le\ell_0}}(\bx)].
\]
For $|C|>\ell_0$, the rounding profile satisfies
\[
  \E[C(\bx)]
  \ge1-\left(\max_\alpha(1-q_\alpha)\right)^{\ell_0+1}
  \ge\rho.
\]
Hence
\[
\begin{aligned}
  L_{\le\ell_0}^\Theta(\Snap_{\eps}(\Phi_{\le\ell_0}))+\rho M_{>\ell_0}
  &\le\E[\Val_{\Phi_{\le\ell_0}}(\bx)]
       +\sum_{C:\,|C|>\ell_0}\E[C(\bx)]\\
  &=\E[\Val_\Phi(\bx)]\le\OPT(\Phi).
\end{aligned}
\]

\paragraph{Step 2. Prove the lower bound.}
Take an optimal assignment $\bx^*$ for $\Phi_{\le\ell_0}$.
For each clause $C$ of length $\ell$, use the record
\[
  c=(\ell;s_1,\ldots,s_\ell;i_1,\ldots,i_\ell;\tau_1,\ldots,\tau_\ell),
  \qquad (s_j,i_j,\tau_j)=(s_j(C),i'(v_j(C)),x^*_{v_j(C)}).
\]
Define $n_c:=\#\{C\in\Phi_{\le\ell_0}:C\text{ has record }c\}$ and write $\alpha_j=(s_j,i_j)$. Recall the clause score
\[
  h_\Theta(c):=\begin{cases}
    q_{\alpha_1},&\ell=1,\\
    1-(1-q_{\alpha_1})(1-q_{\alpha_2}),&\ell=2,\\
    H_\Theta(\alpha_1,\ldots,\alpha_\ell),&3\le\ell\le\ell_0,
  \end{cases}
\]
and set
\[
  o(c):=\1\bigl[\exists j\in[\ell]:\ (s_j,\tau_j)\in\{(+,1),(-,0)\}\bigr].
\]
Then
\[
  L_{\le\ell_0}^\Theta(\Snap_{\eps}(\Phi_{\le\ell_0}))=\sum_c n_ch_\Theta(c),\qquad
  \OPT(\Phi_{\le\ell_0})=\sum_c n_co(c).
\]
Write $\chi(+)=1$, $\chi(-)=-1$, and recall
\[
  g^{\rm lo}_{i,\tau}(c):=
  w_\ell\sum_{\substack{j\in[\ell]:\,i_j=i\\ \tau_j=\tau}}(a_i-\chi(s_j)),
  \qquad
  g^{\rm up}_{i,\tau}(c):=
  w_\ell\sum_{\substack{j\in[\ell]:\,i_j=i\\ \tau_j=\tau}}(\chi(s_j)-b_i).
\]
Constraint~\eqref{eq:snapshot-record-lower} gives nonnegative coefficients $\lambda^{\rm lo}_{i,\tau},\lambda^{\rm up}_{i,\tau}$ such that, for every record $c$,
\[
  h_\Theta(c)-\rho o(c)
  +\sum_{i,\tau}\bigl(
    \lambda^{\rm lo}_{i,\tau}g^{\rm lo}_{i,\tau}(c)
    +\lambda^{\rm up}_{i,\tau}g^{\rm up}_{i,\tau}(c)\bigr)\ge0.
\]
When scaling the weights to integers, divide these coefficients by the same factor; the products $\lambda^{\rm lo}_{i,\tau}w_\ell$ and $\lambda^{\rm up}_{i,\tau}w_\ell$ are unchanged.
For $i\in[L]$ and $\tau\in\bits$, define
\[
  d_{i,\tau}:=\sum_{\substack{v:\,i'(v)=i,\ x_v^*=\tau}}\widetilde d_v,\qquad
  e_{i,\tau}:=\sum_{\substack{v:\,i'(v)=i,\ x_v^*=\tau}}(\widetilde p_v-\widetilde q_v).
\]
The bias bounds and $\widetilde p_v-\widetilde q_v=\widetilde b_v\widetilde d_v$ give
\[
  (a_i-\eps)d_{i,\tau}\le e_{i,\tau}\le(b_i+\eps)d_{i,\tau}.
\]
Thus
\[
  \sum_c n_cg^{\rm lo}_{i,\tau}(c)=a_id_{i,\tau}-e_{i,\tau}\le\eps d_{i,\tau},\qquad
  \sum_c n_cg^{\rm up}_{i,\tau}(c)=e_{i,\tau}-b_id_{i,\tau}\le\eps d_{i,\tau}.
\]
Multiplying the record inequality by $n_c$ and summing gives
\[
\begin{aligned}
  L_{\le\ell_0}^\Theta(\Snap_{\eps}(\Phi_{\le\ell_0}))-\rho\OPT(\Phi_{\le\ell_0})
  &\ge-\sum_{i,\tau}\left(
    \lambda^{\rm lo}_{i,\tau}\sum_c n_cg^{\rm lo}_{i,\tau}(c)
    +\lambda^{\rm up}_{i,\tau}\sum_c n_cg^{\rm up}_{i,\tau}(c)\right)\\
  &\ge-\eps\sum_{i,\tau}(\lambda^{\rm lo}_{i,\tau}+\lambda^{\rm up}_{i,\tau})d_{i,\tau}
  \ge-\gamma_2\eps m,
\end{aligned}
\]
Here, by~\eqref{eq:weighted-total-copies},
\[
  \sum_{i,\tau}d_{i,\tau}=\sum_v\widetilde d_v
  \le m\max_{\ell\in[\ell_0]}\ell w_\ell=O(m),
\]
so we may take
\[
  \gamma_2=\left(\max_{\ell\in[\ell_0]}\ell w_\ell\right)
      \max_{i,\tau}(\lambda^{\rm lo}_{i,\tau}+\lambda^{\rm up}_{i,\tau}).
\]
Finally, $\OPT(\Phi)\le\OPT(\Phi_{\le\ell_0})+M_{>\ell_0}$ gives
\[
\begin{aligned}
  L_{\le\ell_0}^\Theta(\Snap_{\eps}(\Phi_{\le\ell_0}))+\rho M_{>\ell_0}
  &\ge\rho\OPT(\Phi_{\le\ell_0})+\rho M_{>\ell_0}-\gamma_2\eps m\\
  &\ge\rho\OPT(\Phi)-\gamma_2\eps m.
\end{aligned}
\]
\end{proof}

\begin{restatedclosenesslemma}[Closeness of the two snapshots]
With probability at least $1-O(\eps)$ over $\{f_a\}$ and $g$, we have
\[
  \norm{\PsSnap(\Phi_{\le\ell_0})-\Snap_{\eps}(\Phi_{\le\ell_0})}_1
  \le O(\eps m).
\]
\end{restatedclosenesslemma}

We first bound the prefix errors, then count the clauses whose pseudo-biases may give labels different from $(s_j(C),i'(v_j(C)))$.
\begin{proof}[Proof of Lemma~\ref{lem:weighted-snapshot-closeness}]
\begin{samepage}
\leavevmode\par
\paragraph{Step 1. Bound the prefix errors.}
Fix a variable $x_v$ with $\widetilde d_v>0$. Number its copies $r\in[\widetilde d_v]$ by clause arrival and then by their copy index.
For $r>\lfloor\eps\widetilde d_v\rfloor$, the clause $C$ containing that copy satisfies $d_v^{\le C}\ge r>\eps\widetilde d_v$.
If $\widetilde d_v<1/\eps$, all prefix degrees are kept exact.
Otherwise, the range $(\eps\widetilde d_v,\widetilde d_v]$ meets only $O(\eps^{-3}\log(1/\eps))$ degree intervals.
\end{samepage}

Recall $D_0=0$, $D_a=\lfloor(1+\eps^3)^a\rfloor$ for $a\ge1$, and $\kappa=\lceil K\eps^{-7}\rceil$.
Consider one such interval with $2D_a\ge\kappa$; $f_a$ samples all positive copies of $x_v$ independently with probability $\theta_a:=\kappa/(2D_a)$.
The interval width satisfies
\[
  D_{a+1}-D_a\le\eps^3D_a+2=O(\eps^3D_a).
\]
Among the first $D_a+1$ copies of $x_v$, let $p_0$ count the positive copies and let $s_0$ count those selected by $f_a$.
Since $\E[s_0]=\theta_ap_0\le(1+O(\eps^3))\kappa/2$, Bernstein's inequality gives
\[
  |s_0-\theta_ap_0|=O(\eps^3\kappa)
\]
except with probability $e^{-\Omega(\eps^6\kappa)}$.
There are at most $O(\eps^3D_a)$ later copies of $x_v$ with prefix degree at most $D_{a+1}$.
By a Chernoff bound, at most $O(\eps^3\kappa)$ positive copies among them are selected by $f_a$, except with probability $e^{-\Omega(\eps^3\kappa)}$.
Thus, for every $C$ with $D_a<d_v^{\le C}\le D_{a+1}$,
\[
  |s_v^{\le C}-\theta_ap_v^{\le C}|
  \le |s_0-\theta_ap_0|+(s_v^{\le C}-s_0)
       +\theta_a(p_v^{\le C}-p_0)
  =O(\eps^3\kappa).
\]
Together with $\theta_ap_v^{\le C}\le(1+O(\eps^3))\kappa/2$, this bound gives $s_v^{\le C}<\kappa$.
Recall the prefix approximations
\[
  \widetilde d_v^{\le C}=D_a,\qquad
  \widetilde p_v^{\le C}=\frac{2D_a}{\kappa}\min\{s_v^{\le C},\kappa\}.
\]
Multiplying the sampling error by $2D_a/\kappa$ and using the interval-width bound for the total degree gives
\begin{equation}
  |\widetilde p_v^{\le C}-p_v^{\le C}|=O(\eps^3d_v^{\le C}),\qquad
  |\widetilde d_v^{\le C}-d_v^{\le C}|=O(\eps^3d_v^{\le C}).
  \label{eq:weighted-prefix-errors}
\end{equation}
Both errors are zero when $2D_a<\kappa$.
Taking a union bound over these intervals gives both estimates for every clause containing a copy with $r>\lfloor\eps\widetilde d_v\rfloor$, except with probability
\[
  O(\eps^{-3}\log(1/\eps))e^{-\Omega(\eps^6\kappa)}\le\eps^2,
\]
for a sufficiently large $K$ in the choice of $\kappa$.

\paragraph{Step 2. Compare the bucket indices.}
Fix a clause $C$ containing $x_v$ for which both estimates in~\eqref{eq:weighted-prefix-errors} hold.
Using
\[
  \widetilde p_v=p_v^{\le C}+p_v^{>C},\qquad
  \widetilde d_v=d_v^{\le C}+d_v^{>C},\qquad
  \widetilde b_v=2\widetilde p_v/\widetilde d_v-1,
\]
the prefix error bounds give
\[
  \widetilde d_v^{\le C}+d_v^{>C}
  \ge(1-O(\eps^3))\widetilde d_v,
\]
and, since $\widetilde p_v\le\widetilde d_v$,
\[
  \left|
    \frac{2(\widetilde p_v^{\le C}+p_v^{>C})}
         {\widetilde d_v^{\le C}+d_v^{>C}}
    -\frac{2\widetilde p_v}{\widetilde d_v}
  \right|=O(\eps^3).
\]
Clipping cannot increase the difference, so $|\widetilde b_v^{\,C}-b'_v|=O(\eps^3)$.
The bucket index of $\widetilde b_v^{\,C}$ can differ from $i'(v)$ only if $\widetilde b_v+g(v)$ is within $\gamma_4\eps^3$ of an internal bucket endpoint, for a sufficiently large constant $\gamma_4$.
Since $g(v)$ is uniform on an interval of length $2\eps$ and there are $L-1$ such endpoints, this event has probability $O(\eps^2)$.

\paragraph{Step 3. Bound the coordinate distance.}
The exclusions below only count label differences; both snapshots use all clauses of $\Phi_{\le\ell_0}$.
Exclude copies with $r\le\lfloor\eps\widetilde d_v\rfloor$. Their number satisfies
\[
  \sum_v\lfloor\eps\widetilde d_v\rfloor
  \le\eps\sum_v\widetilde d_v
  \le\eps m\max_{\ell\in[\ell_0]}\ell w_\ell=O(\eps m).
\]
Also exclude all copies of a variable for which the positive or total prefix error bound in~\eqref{eq:weighted-prefix-errors} fails at some $C$ containing a copy with $r>\lfloor\eps\widetilde d_v\rfloor$.
Finally, exclude all copies of any variable $x_v$ with $\widetilde b_v+g(v)$ within $\gamma_4\eps^3$ of an internal bucket endpoint.
By Steps 1 and 2, the expected number of copies excluded for prefix errors or proximity to bucket endpoints is at most
\[
  O(\eps^2)\sum_v\widetilde d_v=O(\eps^2m).
\]
Markov's inequality shows that there are only $O(\eps m)$ excluded copies with probability at least $1-O(\eps)$.
Each excluded copy belongs to one clause, so only $O(\eps m)$ clauses contain excluded copies.
All other clauses have the same labels in the two snapshots.
For each clause of length $\ell$ whose labels differ, the sum of the absolute changes in its one-endpoint and two-endpoint coordinates is at most
\[
  2\left(\ell+\binom\ell2\right)\le\ell_0(\ell_0+1).
\]
Summing over the $O(\eps m)$ clauses proves the distance bound.

\end{proof}

\section{Proofs for weighted pseudo-snapshot estimation}
\label{app:quantum-estimation-proofs}

In this section, we prove Lemmas~\ref{lem:nonwrapping-threshold-invariant} and~\ref{lem:good_label_budget} on live labels and numbers of shifts.

\theoremstyle{plain}
\newtheorem*{restatedthresholdlemma}{Lemma~\ref{lem:nonwrapping-threshold-invariant}}
\newtheorem*{restatedoneendpointbudgetlemma}{Lemma~\ref{lem:good_label_budget}}
\theoremstyle{definition}

\subsection{Live labels and prefix counts}
\label{app:one-endpoint-threshold-proof}

Fix a sub-sketch and a label $F(u,t)$ that it queries, where $F$ is $H$ or $E$. For this fixed family and variable, the queried level $t$ is fixed, and at least one ordinary $+1$ update occurs between consecutive queries of $F(u,t)$. Immediately before a query, let $T$ be its live set, let $d_u$ count the ordinary literal-copy updates applied to this chain so far and $c_u$ count its additional $B_a$-shift updates. The elementary shifts inside a $B_a$-shift update are not counted again in $d_u$. Set $c_u=0$ for a family with no additional updates.

\begin{restatedthresholdlemma}
Suppose $1\le t<M$ and, before each query of $F(u,t)$, the sub-sketch has used fewer than $M$ elementary shifts in total. Then, immediately before each such query,
\[
F(u,t)\in T
\quad\Longleftrightarrow\quad
d_u+B_a c_u\ge t.
\]
\end{restatedthresholdlemma}

\begin{proof}[Proof of Lemma~\ref{lem:nonwrapping-threshold-invariant}]
The first $M-1$ elementary shifts use distinct, initially live spare labels. Each shift on $F(u,\cdot)$ inserts a live label at level $1$ and moves the previous labels forward by one. Since the sub-sketch only deletes level $t$ of this chain, its lower levels always satisfy
\[
F(u,r)\in T
\quad\Longleftrightarrow\quad
r\le d_u+B_ac_u
\qquad(1\le r<t).
\]
Before the first query, the same equivalence holds at $r=t$. After a failed query deletes $F(u,t)$, the next query follows at least one update on this chain. If $t>1$, then $d_u+B_ac_u\ge t$ at deletion, so
\[
F(u,t-1)\in T
\xrightarrow{\text{next shift}}
F(u,t)\in T;
\]
if $t=1$, the next shift inserts a fresh spare label. Further shifts keep level $t$ live. If level $t$ has not yet been reached, no query can delete it, and it becomes live precisely when the number of shifts on this chain reaches $t$. This proves the equivalence before every query.
\end{proof}

\subsection{Queried levels and numbers of shifts}
\label{app:one-endpoint-budget-proof}

\begin{restatedoneendpointbudgetlemma}[No-wrap event for one-endpoint sketches]
For each high-degree interval $a$, let $X_a$ count all positive literal copies selected by $f_a$ over the whole stream. For a sufficiently large constant $K_\eps$, define
\[
\mathcal G_1:=
\bigcap_{a:\,D_a\ge\kappa/2}
\left\{X_a\le K_\eps\frac{m}{D_a}\right\}.
\]
Then $\Pr[\mathcal G_1]\ge1-1/160$. Choosing $\gamma_\eps$ in $M=\lceil\gamma_\eps m\rceil$ sufficiently large ensures that every queried level satisfies $1\le t<M$ and, on $\mathcal G_1$, every one-endpoint sub-sketch performs fewer than $M$ elementary shifts over the whole stream.
\end{restatedoneendpointbudgetlemma}

\begin{proof}[Proof of Lemma~\ref{lem:good_label_budget}]
By \eqref{eq:weighted-total-copies},
\[
\sum_v\widetilde p_v\le\sum_v\widetilde d_v
\le m\max_{\ell\in[\ell_0]}\ell w_\ell=O(m).
\]
In a low-degree interval, $D_{a+1}=O_\eps(1)$ and $B_a=O_\eps(1)$, so each copy applies $O_\eps(1)$ shifts and each sub-sketch uses $O_\eps(m)$ shifts.

For a high-degree interval $a$, let $N=O(m)$ be the number of positive literal copies. Then
\[
X_a\sim\mathrm{Bin}\!\left(N,\frac{\kappa}{2D_a}\right),
\qquad
\mathbb E[X_a]=O_\eps(m/D_a).
\]
A Chernoff bound gives
\[
\Pr\!\left[X_a>K_\eps\frac{m}{D_a}\right]
\le\exp\!\left(-c_\eps\frac{m}{D_a}\right),
\]
where $c_\eps$ can be made arbitrarily large by increasing $K_\eps$. There are $O_\eps(1)$ degree intervals in each range $D_a\in[m/2^{r+1},m/2^r]$, and $O_\eps(1)$ intervals with $D_a>m$, since $D_A=O_\eps(m)$. Thus, for sufficiently large $K_\eps$,
\[
\Pr[\mathcal G_1^c]
\le O_\eps(1)\sum_{r\ge0}\exp(-c'_\eps 2^r)
<\frac1{160}.
\]
On $\mathcal G_1$, each high-degree sub-sketch uses at most $O(m)$ ordinary shifts and
\[
B_aX_a
\le(2K_{\mathrm{base}}D_{a+1}+1)K_\eps\frac{m}{D_a}
=O_\eps(m)
\]
additional shifts, since $D_{a+1}\le(1+\eps^3)D_a+2$. This bound holds for every allowed $B_a$, so $\mathcal G_1$ depends only on the sampling bits.

The queried levels are $B_ac+d$ and $B_ac+d+1$ in the low-degree case, with $D_a<d\le D_{a+1}$ and $0\le c\le d$; these are $O_\eps(1)$. In the high-degree case they are
\[
D_a+1,\qquad D_{a+1}+1,\qquad qB_a\quad(1\le q\le\kappa),
\]
all of which are $O_\eps(m)$. A sufficiently large $\gamma_\eps$ makes these positive levels and all the preceding shift bounds strictly less than $M$.
\end{proof}

\end{document}